\documentclass[journal,twocolumn]{IEEEtran}
\pdfoutput=1
\newlength{\figwidth}
\usepackage{xcolor}
\definecolor{links}{rgb}{0.7,0,0}   
\definecolor{urls}{rgb}{0,0,0.8}    
\definecolor{cites}{rgb}{0,0,0.8}   
\usepackage[colorlinks,hyperindex,linkcolor=links,citecolor=cites,urlcolor=urls]{hyperref} 

\usepackage[nosort]{cite}
\usepackage{url}
\usepackage[intlimits]{amsmath}
\usepackage{bbm}
\usepackage{graphicx}
\usepackage{paralist}
\usepackage{fancyref}
\usepackage[stretch=16,shrink=16,step=4]{microtype}
\usepackage{vmr-symbols-vecbold}
\usepackage{standard-macros}
\usepackage{siunitx}
\usepackage{tikz}
\usepackage{pgfplots}
\pgfplotsset{compat=1.18}
\usepackage[font=footnotesize]{caption} 
\usepackage[font=footnotesize]{subcaption}
\usepackage{glossaries}
\glsdisablehyper
\newacronym{rf}{RF}{radio frequency}
\newacronym{cpu}{CPU}{central processing unit}
\newacronym{mimo}{MIMO}{multiple-input multiple-output}
\newacronym{ap}{AP}{access point}
\newacronym{adc}{ADC}{analog-to-digital converter}
\newacronym{vga}{VGA}{variable-gain amplifier}
\newacronym{dsp}{DSP}{digital signal processing}
\newacronym{bb}{BB}{base-band}
\newacronym{bpf}{BPF}{band-pass filter}
\newacronym{osr}{OSR}{oversampling rate}
\newacronym{ue}{UE}{user equipment}
\newacronym{ofdm}{OFDM}{orthogonal frequency-division multiplexing}
\newacronym{lmmse}{LMMSE}{linear minimum mean-square error}
\newacronym{blmmse}{BLMMSE}{Bussgang linear minimum mean-square error}
\newacronym{evm}{EVM}{error vector magnitude}
\newacronym{mr}{MR}{maximum ratio}
\newacronym{zf}{ZF}{zero forcing}
\newacronym{mmse}{MMSE}{minimum mean squared error}
\newacronym{nmse}{NMSE}{normalized mean squared error}
\newacronym{mse}{MSE}{mean squared error}
\newacronym{rof}{RoF}{radio-over-fiber}
\newacronym{mmwave}{mm-wave}{millimeter-wave}
\newacronym{rrh}{RRH}{remote radio head}
\newacronym{csi}{CSI}{channel state information}
\newacronym{crb}{CRB}{Cram\'{e}r-Rao Bound}
\newacronym{snr}{SNR}{signal-to-noise ratio}
\newacronym{sdr}{SDR}{signal-to-dither ratio}
\newacronym{sinr}{SINR}{signal-to-interference-plus-noise ratio}
\newacronym{snn}{SNN}{scaled-nearest neighbor}
\newacronym{rcu}{RCU}{random-coding union}
\newacronym{uatf}{UatF}{use-and-then-forget}
\newacronym{se}{SE}{spectral efficiency}
\newacronym{ml}{ML}{maximum likelihood}
\newacronym{gmi}{GMI}{generalized mutual information}
\newacronym{dft}{DFT}{discrete Fourier transform}
\newacronym{lna}{LNA}{low-noise amplifier}
\newacronym{dmimo}{D-MIMO}{distributed \glsentryshort{mimo}}
\newacronym{dac}{DAC}{digital-to-analog converter}
\newacronym{ga}{GA}{genetic algorithm
}
\newacronym{cdf}{CDF}{cumulative distribution function
}

\newcommand{\xx}[4]{#1_{#2,#3}^{(#4)}}

\newcommand{\nb}{n_\text{b}}
\newcommand{\nd}{n_\text{d}}
\newcommand{\np}{n_\text{p}}

\makeatletter
\def\@IEEEinterspaceratioM{0.265}
\def\@IEEEinterspaceMINratioM{0.1651}
\def\@IEEEinterspaceMAXratioM{0.38}

\def\@IEEEinterspaceratioB{0.31}
\def\@IEEEinterspaceMINratioB{0.19}
\def\@IEEEinterspaceMAXratioB{0.38}
\@IEEEtunefonts
\makeatother
\begin{document}

\title{Distributed Massive MIMO with 1-Bit Radio-over-Fiber Fronthaul: Uplink Spectral
	Efficiency and Power Control}

\author{Alireza Bordbar,~\IEEEmembership{Graduate Student Member,~IEEE}, Anzhong Hu,~\IEEEmembership{Member,~IEEE}, and Giuseppe Durisi,~\IEEEmembership{Senior Member,~IEEE}%
\thanks{Alireza Bordbar and Giuseppe Durisi are with the Department of Electrical Engineering, Chalmers University of Technology, 41296 Gothenburg, Sweden (e-mail: bordbar@chalmers.se; durisi@chalmers.se).}%
\thanks{Anzhong Hu is with Hangzhou Dianzi University, Hangzhou, China (e-mail: huaz@hdu.edu.cn).}%
\thanks{This work was supported in part by the Swedish Foundation for Strategic Research (SSF), under grant FUS21-0004. The simulations were enabled by resources provided by the National Academic Infrastructure for Supercomputing in Sweden (NAISS), partially funded by the Swedish Research Council through grant agreement no. 2022-06725. A conference version of this paper will be presented at the 2026 IEEE International Workshop on Signal Processing and Artificial Intelligence in Wireless Communications (SPAWC).}}

\maketitle
\begin{abstract}
We analyze the uplink spectral
    efficiency achievable in a distributed multiple-input multiple-output (D-MIMO) architecture employing a 1-bit radio-over-fiber fronthaul.
    This architecture eliminates the need for local oscillators at the access points, hence enabling coherent-phase transmission without costly over-the-air synchronization.
    With this fronthaul architecture, the uplink signal at the central processing unit is a dithered, oversampled, and 1-bit quantized version of the pass-band signal received at the access points.
    This makes some of the conventional spectral-efficiency expressions used in the D-MIMO literature not directly applicable for two key reasons: the nonlinearity of the input-output relation, and the practical unavailability of minimum mean square error (MMSE) channel estimates.
    To address this issue, we propose novel achievable rate expressions that do not require MMSE channel estimates and rely on Bussgang decomposition for a linearization of the input-output relation.
    We use these expressions to determine the optimal signal-to-dither ratio (SDR) that maximizes the achievable rates in both single- and multiuser scenarios and to assess the impact of oversampling. We then use one of the proposed achievable-rate expressions to investigate the max-min fairness problem when the access points cannot maintain the optimal SDR because of limitations in the dynamic range.
\end{abstract}

\section{Introduction}
In distributed \glsentrylong{mimo} (\glsentryshort{dmimo})%
\glsunset{mimo}\glsunset{dmimo}, a large number of geographically separated \glspl{ap}, connected to a \gls{cpu} via a fronthaul link, serve several \glspl{ue} over the same time-frequency resources~\cite{Bjornson_hardware}. The main benefit of \gls{dmimo} is that it provides a more uniform quality of service compared with co-located massive \gls{mimo} architectures. A major practical challenge in \gls{dmimo} architectures is to phase-synchronize the local oscillators at the \glspl{ap}, which is necessary for reciprocity-based downlink beamforming. This can be done via over-the-air synchronization techniques, at the cost of signaling overhead~\cite{hamed_real-time_2016, Larsson, Shandi}.
A different approach is considered in~\cite{aabel24-10a, aabel25-11a, Vandierendonck, Kaneko_Kuwabara_Tawa_Maruta_2025}, where \glspl{ap} and the \gls{cpu} are connected via an optical-fiber fronthaul link over which \gls{rf} signals are exchanged, and downconversion is performed at the \gls{cpu}. This eliminates the need for local oscillators at the \glspl{ap} and the corresponding synchronization problem. However, transmitting \gls{rf} signals over optical fiber is challenging because of propagation distortion. Therefore, in~\cite{aabel24-10a, aabel25-11a} a two-level dithered-and-quantized version of the \gls{rf} signals is exchanged over the fronthaul. This two-level waveform is then oversampled at the \gls{cpu} via high-speed 1-bit \glspl{adc}.
Following~\cite{aabel24-10a, aabel25-11a}, we refer to this architecture as \emph{\gls{dmimo} with 1-bit \gls{rof} fronthaul}.

Maintaining an optimal \gls{sdr} is critical for maximizing the performance of this architecture~\cite{hu24-11a}. In the testbed described in~\cite{aabel24-10a,aabel25-11a}, this is accomplished by using a \gls{vga} to normalize the received power to a desired level before adding a dither signal of fixed power. Unfortunately, the \gls{vga}'s limited dynamic range prevents it from maintaining the optimal \gls{sdr} under certain operating conditions.

In this paper, we characterize the uplink spectral efficiency achievable with the \gls{dmimo} with 1-bit \gls{rof} fronthaul architecture by deriving bounds on achievable rates. We then use these bounds to study the effect of oversampling at the \gls{cpu} on the spectral efficiency and to determine the \gls{sdr} that maximizes the achievable rates. 
Furthermore, we investigate how to mitigate via \gls{ue} power control the negative effects of the \gls{vga}'s limited dynamic range on the achievable rates.

\subsection{State of the Art}
\paragraph*{\Gls{evm} and \gls{mmse} analyses}
The uplink performance of~\gls{dmimo} systems with 1-bit \gls{rof} fronthaul and the effect of the \gls{vga}'s dynamic range on the \gls{evm} have been studied via measurements in~\cite{aabel24-10a, aabel25-11a}. Additionally, in~\cite{hu24-11a}, an \gls{evm} analysis is provided, which shows that \gls{dmimo} with 1-bit \gls{rof} fronthaul outperforms co-located \gls{mimo}, despite the 1-bit-induced nonlinearity. In~\cite{bordbar24-03a}, a data-driven channel estimator that outperforms the \gls{blmmse} channel estimator~\cite{wan_generalized_2019} is proposed. Furthermore, 
 the proposed data-driven estimator is shown to be more robust to hardware limitations and impairments than the \gls{blmmse} channel estimator.

 \paragraph*{Achievable rates in quantized \gls{mimo} systems}
 The two main challenges in analyzing the achievable rate in 1-bit \gls{mimo} systems are the nonlinearity of the input-output relation and the unavailability of \gls{mmse} channel estimates, since the \gls{mmse} estimator is difficult to implement. The first problem can be overcome by using Bussgang decomposition~\cite{bussgang52a}. However, the unavailability of the \gls{mmse} channel estimates remains a major challenge as it precludes the use of many conventional achievable-rate bounds in the massive \gls{mimo} literature, such as~\cite[Eq.~(4.2)]{Bjornson_hardware} and~\cite[Eq.~(5.25)]{demir_foundations_2021}.

Several works have used Bussgang decomposition to derive achievable-rate expressions or approximations for co-located \gls{mimo} architectures with low-resolution quantization of baseband signal samples. In~\cite{Jacobsson}, approximations for achievable rates for finite-cardinality and Gaussian inputs are proposed for 1-bit \glspl{adc}. In~\cite{Jacobsson_precoding}, lower bounds on the achievable downlink rates for 1- and few-bit quantizers are derived for an auxiliary model in which the quantization noise is replaced with independent Gaussian noise. 
In~\cite{Bjornson-hardware} and~\cite{Mollen}, it is shown that the spectral efficiency with 1-bit \glspl{adc} approaches that of unquantized systems as the number of antennas grows. In~\cite{Liang_2016}, the \gls{gmi} is used to analyze the achievable data rates of a mixed-\gls{adc} architecture. Finally, in~\cite{Zhang}, the uplink achievable rate for cell-free massive \gls{mimo} with low-resolution \glspl{adc} is studied by modeling the quantization distortion as Gaussian. It is important to remark that~\cite{Jacobsson, Jacobsson_precoding, Bjornson-hardware, Mollen, Liang_2016, Zhang} consider the quantization of baseband signals, whereas in the architecture analyzed in the present paper, the 1-bit quantizer is applied to a dithered version of the \gls{rf} signals.

In agreement with the naming conventions in~\cite{kramer_information_2023}, we shall call the bounds in~\cite[Eqs.~(31), (38)]{Jacobsson},~\cite[Eq.~(25)]{Jacobsson_precoding},~\cite[Eq.~(17)]{Bjornson-hardware},~\cite[Eq.~(77)]{Mollen},~\cite[Eq.~(4.2)]{Bjornson_hardware}, and~\cite[Eq.~(5.25)]{demir_foundations_2021} \emph{reverse-model bounds}, since they are obtained by bounding the mutual information via a decomposition involving the conditional distribution of the input signal given both output and channel estimate. In contrast, the \gls{gmi} bound proposed in~\cite[Eqs.~(29), (30)]{Liang_2016} is an example of a \emph{forward-model bound}, as it relies on an auxiliary conditional distribution (or, more generally, a conditional metric) of the output signal given the input and channel estimate.
\paragraph*{\gls{ue} power control}
A common framework to optimize \gls{ue} power control in \gls{dmimo} is
  max-min fairness, in which the goal is to maximize the lowest
  achievable rate among \glspl{ue} under a per-\gls{ue} power constraint.
  For conventional \gls{dmimo} systems with infinite-precision
  \glspl{adc}, this problem is well studied. Existing solutions include
  fixed-point algorithms based on monotonic functions of the
  signal-to-interference ratio~\cite{yates, hong2014unified}, or
  iterative bisection search over the target rate~\cite{boyd2004convex,
  demir_foundations_2021}. To reduce the computational burden,
  meta-heuristics~\cite{conceicao2022maxmin} and learning-based
  strategies~\cite{chafaa2025transformer} have also been proposed.
  However, in \gls{dmimo} with 1-bit \gls{rof} fronthaul architectures, the
  dependence of the achievable rate on the transmitted power is more
  complex, as such quantity affects also the \gls{vga}'s operating point,
  the Bussgang gain, and the quantization-noise covariance matrix. This
  complex dependence breaks the monotonicity properties often assumed in
  the literature, rendering conventional power-control schemes
  inapplicable.

\subsection{Contributions}
The main contributions of this paper are summarized as follows:

\begin{enumerate}

\item We derive a lower bound on mutual information, which we refer to as the \emph{reverse-channel
bound}, using
Bussgang decomposition and the reverse-model bounding
technique. This bound generalizes the ones proposed
in~\cite[Eq.~(4.2)]{Bjornson_hardware}
and~\cite[Eq.~(5.25)]{demir_foundations_2021},
since it holds for arbitrary channel estimation algorithms.
For the special case in which the receiver has access to the
first moment of the effective channel, the proposed bound
reduces to a generalization of the \gls{uatf}
bound~\cite[Eq.~(5.8)]{demir_foundations_2021}
to \gls{dmimo} systems with 1-bit \gls{rof} fronthaul.

\item To overcome the computational challenges associated
with the reverse-channel bound, we derive a significantly
simpler achievable-rate expression based on Bussgang
decomposition and the \gls{gmi}. The resulting bound
relies on the \gls{snn} decoding rule and holds for
arbitrary channel estimation algorithms.

\item Via numerical simulations conducted for realistic
system parameters, we evaluate the impact of the
oversampling rate and of the \gls{sdr} on
the achievable rates. Our results illustrate that the
reverse-channel bound yields the highest achievable-rate
predictions, but becomes computationally difficult to
evaluate at oversampling the rates used in current
testbeds~\cite{aabel24-10a,aabel25-11a}. In contrast,
the \gls{gmi} and \gls{uatf} bounds remain computationally
tractable.
\item We investigate the problem of max-min fairness in the presence of \glspl{ap} equipped with \glspl{vga} of
finite dynamic range. Specifically, we use the \gls{uatf} bound in combination with a standard
  nonlinear optimizer to determine the optimal \gls{ue} power levels for the
  max-min fairness problem. Interestingly, our results show
  that setting the power levels so that the \gls{vga} at each \gls{ap} operates
  within its dynamic range is not necessarily optimal. Indeed, doing so
  may reduce the benefits obtainable via macro diversity. 
\end{enumerate}

Compared with the conference version of this paper~\cite{bordbar26-spawc}, we provide the proofs of the proposed bounds, which were omitted in~\cite{bordbar26-spawc} because of space limitations. Furthermore, we present more comprehensive simulation results. Finally, we formulate the max-min fairness problem in the presence of \glspl{ap} with \glspl{vga} of finite dynamic range, solve this problem via a standard nonlinear optimizer to determine the optimal \gls{ue} transmit powers, and present insights for practical deployment scenarios.

The rest of this paper is structured as follows. In Section~\ref{sec:system-model}, we describe the \gls{dmimo} with  1-bit \gls{rof} system model. In Section~\ref{sec:section3}, we derive the reverse-channel, \gls{uatf}, and \gls{gmi} bounds on the achievable rates for a simpler linear channel model. Then, we apply in  Section~\ref{sec:sec4} these bounds to \gls{dmimo} systems with 1-bit \gls{rof}, after Bussgang linearization. In Section~\ref{sec:simulation-results}, we provide simulation results to assess the impact of oversampling and determine the optimal \gls{sdr}. We also formulate the max-min fairness problem for the case of \glspl{vga} with finite dynamic range, and evaluate the effect of the \glspl{vga}' finite dynamic range on the achievable rates. We conclude the paper in~Section~\ref{sec: conclusion}.

\paragraph*{Notation}
Scalars are denoted by italic letters, vectors by bold lowercase letters, matrices by bold uppercase letters, and sets by calligraphic letters. The $N \times 1$ zero vector is denoted by $\veczero_N$. We denote by $\jpg(\veczero_N , \matC)$ the distribution of an $N$-dimensional circularly symmetric complex-valued Gaussian vector with covariance matrix $\matC$. Similarly, we denote by $\setN(\veczero_N , \matC)$ the distribution of an $N$-dimensional real-valued Gaussian vector with covariance matrix $\matC$.
The sets of real and complex numbers are denoted by $\reals$ and $\complexset$, respectively.
The transpose, Hermitian transpose, and complex conjugate are denoted by $\tp{(\cdot)}$, $\herm{(\cdot)}$, and $\conj{(\cdot)}$, respectively.
We use $\Re\{\cdot\}$ for the real part, $\Ex{}{\cdot}$ for expectation, $\distas$ for equality in distribution, $\diag(\cdot)$ for a diagonal matrix, and $\kron$ for the Kronecker product. We denote by $\delta[n]$ the discrete-time Kronecker delta, i.e., $\delta[0]=1$ and $\delta[n]=0$ for $n\neq 0$.
The operation $\vectorize{\matA}$ stacks the columns of a matrix $\matA$ into a vector.
The function $\mathrm{sgn}(\cdot)$ denotes the sign function and is applied element-wise when its input is a vector.
The identity matrix of size $B$ is denoted by $\matI_B$. We use the standard Landau notation $\mathcal{O}(\cdot)$ to describe asymptotic
  scaling; specifically, $f(x)=\mathcal{O}(g(x))$ means that there exist constants
  $C>0$ and $x_0$ such that $|f(x)|\le C|g(x)|$ for all $x\ge x_0$.

\section{System Model}\label{sec:system-model}
\begin{figure}[t]
	\centerline{{\includegraphics[width=\figwidth]{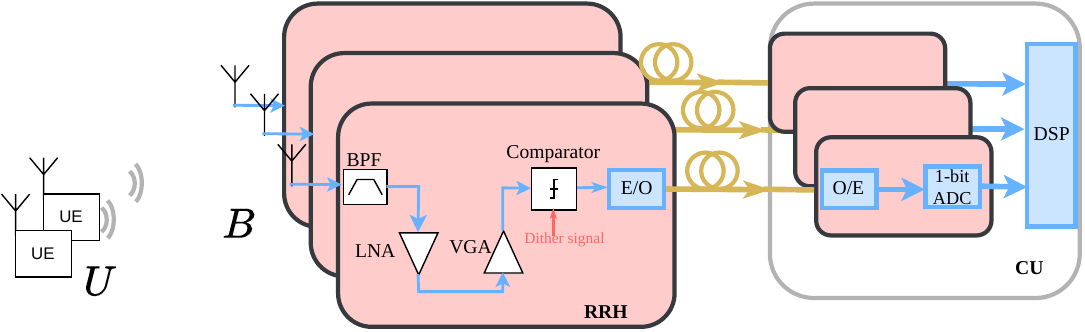}}}
	\caption{A \gls{dmimo} system with 1-bit \gls{rof} fronthaul, in which $U$ single-antenna \glspl{ue} communicate with $B$ \glspl{ap} connected to a \gls{cpu}. The received analog \gls{rf} signals at the $B$ \glspl{ap} are passed through a band-pass filter (BPF), an \gls{lna} and a \gls{vga}. The resulting signal is then compared with a dither signal to generate a two-level analog \gls{rf} waveform. The waveform is converted to the optical domain via an electrical-to-optical (E/O) converter, transmitted to the \gls{cpu} via an optical fiber, and converted back to the electrical domain. The \gls{cpu} samples the two-level  waveform with 1-bit \glspl{adc} to generate 1-bit digital \gls{rf} signals, and then performs signal processing at the digital signal processor (DSP).}
	\label{fig:system}
\end{figure}

We consider the \gls{dmimo} architecture with 1-bit \gls{rof} fronthaul demonstrated
in~\cite{sezgin21-02a,aabel24-10a,aabel25-11a} and depicted in Fig.~\ref{fig:system}.
In the uplink, the received \gls{rf} signal at the antenna port of each single-antenna \gls{ap} is band-pass filtered, amplified, and
then passed through a \gls{vga}.
The output of the \gls{vga} is then injected into a comparator, whose other port receives a suitably
designed dither signal, produced at the \gls{cpu} and sent to each \gls{ap} via the downlink
fronthaul link.
The resulting two-level signal is converted into the optical domain and transmitted over an optical
fiber to the \gls{cpu}. Note that, with this waveform choice, the impact of distortions occurring over the optical fiber is minimized, which makes the design robust.
At the \gls{cpu}, the received signal is converted back to the electrical domain and then digitized
via a 1-bit \gls{adc} operating at a frequency much larger than the signal bandwidth.
As an example, in the testbed described in~\cite{aabel25-11a} a signal of roughly $\qty{80}{MHz}$
centered at $\qty{2.4}{GHz}$ is sampled at $\qty{25}{GS/s}$ at the \gls{ap}.
This yields an extremely large oversampling rate.
Note that the 1-bit \gls{adc} operates directly on the \gls{rf} signal: no analog down-conversion
is performed at the \gls{ap}, which eliminates the need for a local oscillator and the resulting
phase noise.
We assume that down-conversion to \gls{bb} is performed digitally at the \gls{cpu} after sampling
and 1-bit quantization.

As shown in~\cite{hu24-11a}, dithering and oversampling are critical to overcome the distortion
introduced by the 1-bit quantizer.
The role of the \gls{vga} is to ensure that the power of the received signal is scaled so as to make
dithering effective.

Throughout, we focus on the transmission of a \gls{rf} signal of bandwidth $W$, centered at a carrier
frequency $f\sub{c}$.
We let $f\sub{s}$ denote the sampling frequency of the \gls{adc}, and assume that $f\sub{s}\geq
	2f\sub{c} + W$.

\subsection{Uplink Data Transmission}
\label{sec: 2A}
Following~\cite{hu24-11a}, we fix a time duration $T=N/f\sub{s}$, where the integer $N$ denotes the number of samples per
\gls{ap} collected at the \gls{cpu} over the time duration $T$ after 1-bit quantization.
At each discrete time $n$, we gather the $B$ samples from the $B$ \glspl{ap} in a vector
$\vecz_n\supp{rf} \in \{-1,1\}^{B}$, and express this vector as
\begin{equation}
	\label{znURF}
	\vecz_n\supp{rf}=\mathrm{sgn}(\vecy_n\supp{rf}+\vecd_n).
\end{equation}
Here, $\vecy_n\supp{rf}$ denotes the $n$th sample of the received \gls{rf} signal after
band-pass filtering and amplification, and $\vecd_{n}$ denotes the dither signal.
For analytical tractability, we assume throughout the paper that $\vecd_{n}\distas
	\setN(\veczero_{B}, \matD/2)$, where $\matD = \diag (E\sub{d,1}, \dots, E\sub{d,B})$, and that $\vecd_{n}$ is independent across $n$ and
independent of the received signal $\vecy_n\supp{rf}$.
 Note that, to keep the notation compact, we have chosen
  to absorb the amplitude normalization performed by the \glspl{vga} on
  $\vecy_n\supp{rf}$ in the variance of the dither signal, which we
  allow to be \gls{ap}-dependent.

The received signal is modeled as
\begin{equation}
	\label{xnBB}
	\vecy_n\supp{rf}=\sqrt{2}\Re\lefto\{\vecy_n\supp{bb}e^{i2\pi(f\sub{c}/f\sub{s})n}\right\},
\end{equation}
where $\vecy_n\supp{bb}$ stands for the complex envelope.
To model this quantity, and, specifically, account for the correlation introduced by oversampling, we
let $S=WT$ and assume for simplicity that $S$ is an odd integer.
We then define the set
\begin{equation}
	\setS=\{0,1,\dots,(S-1)/2\}\union \{N-(S-1)/2,\dots,N-1\},
\end{equation}
and model $\vecy_n\supp{bb}$ via an inverse discrete Fourier transform as
\begin{equation}
	\label{xnBBaa}
	\vecy_n\supp{bb}=\frac{1}{\sqrt{N}}\sum_{k\in\setS}\left(\bar{\matH}_k\bar{\vecs}_k+\bar{\vecw}_k\right)
	e^{i2\pi\frac{k}{N}n}.
\end{equation}
Here, $\bar{\vecs}_k \in \complexset^{U}$ denotes the frequency-domain transmitted symbols from the
$U$ single-antenna \glspl{ue} with $\bar{\vecs}_k\distas\jpg(\veczero_{U}, \matP_k)$, where
\begin{equation}
	\matP_k = \diag(\rho_{1,k},\dots,\rho_{U,k})
\end{equation}
collects the transmit power used by the \glspl{ue} on subcarrier~$k$.
Furthermore, $\bar{\vecw}_k\distas\jpg(\veczero_{B}, N_{0}\matI_{B})$ is the band-limited additive Gaussian noise,
and $\bar{\matH}_k\in \complexset^{B\times U}$ denotes the channel frequency response.
Throughout the paper, we assume that
$\bar{\vech}_k=\vectorize{\bar{\matH}_k}\distas\jpg(\veczero_{U B},\matL)$, where
$\matL\in\reals^{UB\times UB}$ is a diagonal matrix containing the large-scale fading coefficients
between the \glspl{ap} and the \glspl{ue}.
Specifically, if $\beta_{b,u}$ denotes the large-scale fading coefficient between \gls{ap}~$b$
and \gls{ue}~$u$, then, 
\begin{equation}
  	\matL=\diag(\matL_1,\ldots,\matL_U),
  \end{equation}
with 
\begin{align}
	\label{eq: pathloss_label}
	\matL_u=\diag(\beta_{1,u},\ldots,\beta_{B,u}).
\end{align}
Finally, we assume that to recover the transmit signal, the \gls{cpu} performs digital down-conversion, followed by a Fourier transform and
filtering.
This yields
\begin{equation}
	\label{zkUBB}
	\bar{\vecz}_k\supp{bb}={\sqrt{\frac{2}{N}}}\sum_{n=0}^{N-1}
	\vecz_n\supp{rf}e^{-i2\pi(k/N+f\sub{c}/f\sub{s})n}\in\complexset^{B}
\end{equation}
for $k\in \setS$.

\subsection{Linearization via Bussgang Decomposition}
\label{sec: twoB}
The input--output relation in~\eqref{znURF} is nonlinear. Since the input vector $\vecq_n=\vecy_n\supp{rf}+\vecd_n \in \reals^B$ is conditionally Gaussian given the channel matrices $\{\bar{\matH}_k\}$, we can apply Bussgang decomposition \cite{bussgang52a} to linearize~\eqref{znURF} as
\begin{align}
	\label{eq: znrfBussgang}
	\vecz_n\supp{rf}=\matG\vecq_n+\vece_n\supp{rf},
\end{align}
where $\vece_n\supp{rf}$ is the quantization noise that is uncorrelated with the input $\vecq_n$, and $\matG$ is the Bussgang gain matrix given by~\cite{hu24-11a}
\begin{align}
	\matG = \sqrt{\frac{2}{\pi}} \diag\lefto(\matR_\vecq [0] \right)^{-1/2}.
	\label{eq: Bussgang gain matrix definition}
\end{align}
Here, $\matR_\vecq[n]$ is the autocovariance matrix of $\vecq_n$, given by
\begin{equation}
\begin{split}
	\matR_\vecq[n]
	&= \matR_{\vecy\supp{rf}} [n]
	+ \frac{\matD}{2} \delta [n],
\end{split}
\end{equation}
with 
\begin{equation}
	\label{eq: ryrf}
\begin{split}
	\matR_{\vecy\supp{rf}} [n]
	&= \Ex{}{\vecy\supp{rf}_m\herm{(\vecy\supp{rf}_{n-m})}} \\
	&= \frac{1}{N}\,\Re\Bigg\{\sum_{k \in \setS}
	\matQ_k e^{i 2 \pi (f\sub{c} / f\sub{s} + k/N)n} \Bigg\},
\end{split}
\end{equation}
where $\matQ_k=\bar{\matH}_k \matP_k \herm{\bar{\matH}_k} + N_0 \matI_B$.

Substituting~\eqref{xnBB} and~\eqref{xnBBaa} into \eqref{eq: znrfBussgang}, we obtain
\begin{align}
	\label{eq: zkbbBussgangdecomposition}
	\bar{\vecz}_k\supp{bb}=\matG\bar{\matH}_k\bar{\vecs}_k
	+\matG\bar{\vecw}_k+\matG\bar\vecd_k+\bar\vece_k,
\end{align}
where $\bar\vecd_k\in\complexset^{B}$ and $\bar\vece_k\in\complexset^{B}$ are defined as
\begin{IEEEeqnarray}{rCl}
	\label{eq: dk}
	\bar\vecd_k&=&\sqrt{\frac{2}{N}}\sum_{n=0}^{N-1}\vecd_ne^{-i2\pi(k/N+f\sub{c}/f\sub{s})n},\\
	\label{eq: ek}
	\bar\vece_k&=&\sqrt{\frac{2}{N}}\sum_{n=0}^{N-1}\vece_n\supp{rf}e^{-i2\pi(k/N+f\sub{c}/f\sub{s})n},
\end{IEEEeqnarray}
with covariance matrices
\begin{align}
	\Ex{}{\bar{\vecd}_k \, \herm{\bar{\vecd}_k}} & = \matD,                                                                                                 \\
	\Ex{}{\bar{\vece}_k \, \herm{\bar{\vece}_k}} & = \matC_{\bar{\vece}_k} = 2 \sum_{n = 0}^{N-1} \matR_{\vece}[n] \, e^{-i2 \pi (k/N + f\sub{c}/f\sub{s}) n}.
\end{align}
Here, $\matR_{\vece}[n] \in \reals^{B \times B}$ is given by
\begin{align}
	\matR_\vece [n] = \matR_{\vecz\supp{rf}} [n] - \matG \matR_\vecq[n] \matG.
\end{align}
For a 1-bit quantizer, $\matR_{\vecz\supp{rf}} [n]$ can be computed using Van Vleck's arcsine law~\cite{van-vleck66a}, which yields
\begin{equation}
	\matR_{\vecz\supp{rf}} [n]
	= \frac{2}{\pi} \arcsin\!\Bigg(
	\matJ^{-1/2 } \, \matR_\vecq[n]\, \matJ^{-1/2}\Bigg),
\end{equation}
with $\matJ = \diag(\matR_\vecq [0])$.
\subsection{Acquiring Channel State Information}
  We consider a system in which channel state information at the \glspl{ap} is
  acquired via uplink pilot transmission. In what follows, we adapt 
 the expressions obtained in Section~\ref{sec: 2A} and Section~\ref{sec: twoB}  to  pilot transmission. Specifically, we will provide a model for
  the input-output relation in the pilot-transmission phase. We will
  then linearize it via Bussgang decomposition, and finally, derive a
  \gls{blmmse} channel estimator that generalizes to \gls{dmimo} systems with 1-bit
  \gls{rof} the one derived in~\cite{Li} for the 1-bit baseband
  \gls{mimo} case. We assume that during the  pilot phase, all $U$
\glspl{ue}  transmit simultaneously  pilot sequences over $\np$  symbols to the
\glspl{ap}. Let $\bar{\bm{\Phi}}_k\in\complexset^{U\times n\sub{p}}$
  contain the $n\sub{p}$ pilot symbols transmitted by the $U$ \glspl{ue}
  on subcarrier $k \in \setS$. We assume that the pilot
  sequences are orthogonal, which implies that
  $\bar{\bm{\Phi}}_k\herm{\bar{\bm{\Phi}}_k}=\matI_U$. For convenience, we
  group the received samples 
  corresponding to the $n$th time sample of the $n\sub{p}$ symbols over all $B$ \glspl{ap} in the matrix

\begin{equation}
	\label{YnRFP}
	\matY_n\supp{rf}=\sqrt{2}\Re\lefto\{\matY_n\supp{bb}e^{i2\pi(f\sub{c}/f\sub{s})n}\right\}\in\reals^{B\times n\sub{p}},
\end{equation}
where 
\begin{equation}
\begin{split}
	\matY_n\supp{bb}
	&= \frac{1}{\sqrt{N}}
	\sum_{k\in\setS}\bar{\matH}_k\matP_k^{1/2}\bar{\bm{\Phi}}_k
	e^{i2\pi\frac{k}{N}n} \\
	&\quad{}+ \frac{1}{\sqrt{N}}\sum_{k\in\setS}\bar{\matW}_k
	e^{i2\pi\frac{k}{N}n}.
\end{split}
\label{YnBB}
\end{equation}
Here, $\bar{\matW}_k\in\complexset^{B\times n\sub{p}}$ is the band-limited Gaussian noise. 
It turns out convenient to vectorize the received signal $\matY_n\supp{rf}$ as
\begin{equation}
\begin{split}
	\label{YnuRF2}
	\vecy_n\supp{rf,p}
	&= \vectorize{\matY_n\supp{rf}} \\
	&= \sqrt{\frac{2}{N}}\Re\lefto\{\sum_{k\in\setS}
	\left(\tilde{\bm{\Phi}}_k\bar{\vech}_k
	+\bar{\vecw}_k\right)\right.\\
	&\hspace{5.4em}\left.
	{}\times e^{i2\pi(k/N+f\sub{c}/f\sub{s})n}\right\}
	\in\reals^{n\sub{p} B},
\end{split}
\end{equation}
where $\tilde{\bm{\Phi}}_k=\tp{\left(\matP_k^{1/2}\bar{\bm{\Phi}}_k\right)}\kron \matI_B\in\complexset^{n\sub{p} B\times UB}$, and $\bar{\vecw}_k=\vectorize{\bar{\matW}_k}\in\complexset^{n\sub{p} B}$.

Similar to~\eqref{znURF}, we express the quantized signal at the output of the \gls{adc} as
\begin{equation}
	\label{znURFP}
	\vecz_n\supp{rf,p}=\mathrm{sgn}(\vecq_n\supp{p})\in\reals^{n\sub{p} B},
\end{equation}
where 
\begin{equation}
	\label{qndef}
	\vecq_n\supp{p}= \vecy_n\supp{rf,p}+\vecd_n\supp{p}\in\reals^{n\sub{p} B},
\end{equation}
and 
$\vecd_n\supp{p}\distas\setN(\veczero_{\np B},\matD\sub{p})$, where
$\matD\sub{p}=\matI_{\np}\kron\matD$, is the dither signal. Then, the \gls{cpu} processes the signal and obtains
\begin{equation}
	\label{ZkBBP1}
	\bar{\vecz}_k\supp{bb,p}={\sqrt{\frac{2}{N}}}\sum_{n=0}^{N-1}
	\vecz_n\supp{rf,p}e^{-i2\pi(k/N+f\sub{c}/f\sub{s})n}\in\complexset^{n\sub{p} B}.
\end{equation}
Note that since $\vecq_n\supp{p}$ is Gaussian distributed, we can apply Bussgang's theorem~\cite{bussgang52a} to~\eqref{znURFP}. Specifically, proceeding as in Section~\ref{sec: twoB}, we rewrite~\eqref{znURFP} as
\begin{equation}
	\label{znURFP1}
	\vecz_n\supp{rf,p}=\matG\sub{p}\vecq_n\supp{p}+\vece_n\supp{rf,p},
\end{equation}
where $\vece_n\supp{rf,p}\in\reals^{n\sub{p} B}$ is the distortion term that is uncorrelated with $\vecq_n\supp{p}$, and $\matG\sub{p}\in\reals^{n\sub{p} B\times n\sub{p} B}$ is given by
\begin{IEEEeqnarray}{rCl}
	\label{GP1}
	\matG\sub{p}
	&=& \sqrt{\frac{2}{\pi}}\diag\Bigg(
	\frac{1}{N}\Re\Bigg\{
	\sum_{k\in\setS}\tilde{\bm{\Phi}}_k\matL{\herm{\tilde{\bm{\Phi}}_k}}
	\Bigg\}
	\nonumber\\
	&&\hspace{4em}
	+\frac{S}{N}N_0\matI_{n\sub{p} B}
	+\frac{1}{2}\matD\sub{p}\Bigg)^{-\frac{1}{2}}.
\end{IEEEeqnarray}

Substituting~\eqref{YnuRF2} and~\eqref{qndef} into~\eqref{znURFP1} and then~\eqref{znURFP1} into~\eqref{ZkBBP1}, we can express $\bar{\vecz}_k\supp{bb,p}$ as
\begin{equation}
	\label{ZkBBP2}
	\bar{\vecz}_k\supp{bb,p}=\matG\sub{p}\left(\tilde{\bm{\Phi}}_k\bar{\vech}_k+\bar{\vecw}_k+\bar\vecd_k\supp{p}\right)+\bar\vece_k\supp{p},
\end{equation}
where $\bar\vecd_k\supp{p}\in\complexset^{n\sub{p} B}$ and $\bar\vece_k\supp{p}\in\complexset^{n\sub{p} B}$ are defined similarly to~\eqref{eq: dk} and~\eqref{eq: ek}.

The \gls{blmmse} channel estimator~\cite{Li} for the input-output relation in~\eqref{ZkBBP2} is given by
\begin{equation}
	\label{hBLM1}
	\hat{\vech}_k
	= \matL\herm{\tilde{\bm{\Phi}}_k}
	\matC_{\tilde{\vecz}_k\supp{bb,p}}^{-1}
	\tilde{\vecz}_k\supp{bb,p}
	\in\complexset^{UB},
\end{equation}
where
\begin{equation}
	\tilde{\vecz}_k\supp{bb,p}
	= \matG\sub{p}^{-1}\bar{\vecz}_k\supp{bb,p},
\end{equation}
and 
\begin{equation}
	\label{eq:bussgang-lmmse-pilot-covariance}
	\matC_{\tilde{\vecz}_k\supp{bb,p}}
	=
	\tilde{\bm{\Phi}}_k\matL\herm{\tilde{\bm{\Phi}}_k}
	+ N_0\matI_{n\sub{p} B}
	+ \matD\sub{p}
	+ \matG\sub{p}^{-1}\matC_{\bar{\vece}_k\supp{p}}\matG\sub{p}^{-1}.
\end{equation}
Here,
$\matC_{\bar{\vece}_k\supp{p}}=\Ex{}{\bar{\vece}_k\supp{p}\herm{(\bar{\vece}_k\supp{p})}}$ is the covariance matrix of the 
frequency-domain Bussgang quantization error on  subcarrier~$k$.

We note that computing the \gls{mmse} channel estimate for the input-output relation in~\eqref{ZkBBP2} is not 
possible, since the distribution of $\bar\vece_k\supp{p}$ is not known in closed form. This makes computing 
$\Ex{}{\bar{\vech}_k|\bar{\vecz}_k\supp{bb,p}}$ difficult. Thus, achievable-rate expressions that require \gls{mmse} channel estimates such as~\cite[Thm.~4.1]{Bjornson_hardware}, are typically not applicable to \gls{dmimo} systems with 1-bit \gls{rof} fronthaul.

\subsection{\glsentryshort{vga}'s Dynamic Range}
\label{sec:vga-model}
We next discuss how to capture the finite dynamic range of the \gls{vga} in
  our system model, in which we have chosen to absorb the amplitude
  normalization performed by the \gls{vga} in the variance of the dither
  signal. Let $P_b\supp{rf}$ denote the average power of the signal received at \gls{ap} $b$
  during the data-transmission phase:
\begin{equation}
\begin{split}
	P_b\supp{rf} &= \lefto[N\Ex{\{\bar{\matH}_k\}}{\matR_{\vecy\supp{rf}} [0]}\right]_{b,b} \\
	&= \sum_{k\in\setS}
	\left(\sum_{u=1}^{U} \rho_{u,k}\beta_{b,u}+N_0\right).
\end{split}
\label{eq:average-rf-power-ap}
\end{equation}
To maintain a desired \gls{sdr} target $\gamma^\star$, we should ideally set the dither power at each \gls{ap} to
\begin{equation}
	E\sub{d,b} = \frac{P_b\supp{rf}}{\gamma^\star}, \quad b=1,\ldots,B.
	\label{eq: edruleidealvga}
\end{equation}
Unfortunately, the finite dynamic range of the \gls{vga} implies that we can enforce~\eqref{eq: edruleidealvga} only over a limited range of $P_b\supp{rf}$ values. To see why, let us consider the \gls{vga} hardware model detailed in~\cite{aabel25-11a}. Let\footnote{
We use the convention that whenever we add the superscript \qty{}{dBm} or \qty{}{dB} to an already-defined quantity, we mean that this quantity is measured in \qty{}{dBm} or \qty{}{dB}.
} $P_b\supp{rf,dBm} = 10 \log_{10}\left(P_b\supp{rf}\right)$ and $G\sub{lna}$ be the gain in dB of the \gls{lna}, which is equal to \qty{33}{dB}. The average power (in \qty{}{dBm}) at the output of the \gls{lna} is given by 
\begin{equation}
	\label{eq: pblna}
	P_b\supp{lna} = P_b\supp{rf,dBm}+G\sub{lna}.
\end{equation}
In the testbed developed in~\cite{aabel25-11a}, the \gls{vga} is able to normalize the average power at the output of the \gls{vga} to \qty{-31}{dBm} only when $P_b\supp{lna} \in (\qty{-41}{dBm}, \qty{4}{dBm})$. More specifically, the average power after the \gls{vga} is $P_b\supp{lna} + \alpha(P_b\supp{lna})$, where the \gls{vga} gain $\alpha(P_b\supp{lna})$ is given by
\begin{equation}\label{eq:vga-gain-model}
	\alpha(P_b\supp{lna}) =
	\begin{cases}
		\qty{10}{dB}, & P_b\supp{lna} < \qty{-41}{dBm},\\
		\qty{-35}{dB}, & P_b\supp{lna} > \qty{4}{dBm},\\
		-P_b\supp{lna}-\qty{31}{dB}, & \text{otherwise}.
	\end{cases}
\end{equation} 
To capture the \gls{vga}'s finite dynamic range, we replace $\gamma^\star$ in~\eqref{eq: edruleidealvga} with $\gamma$, defined as 
\begin{equation}
\label{eq:actual-signal-to-dither-ratio}
\gamma\supp{dB} =
\begin{cases}
	\gamma^{\star,\mathrm{dB}} + P_b\supp{lna} + 41, 
& P_b\supp{lna} < \qty{-41}{dBm} \\[0.5ex]
\gamma^{\star,\mathrm{dB}} + P_b\supp{lna} - 4, 
& P_b\supp{lna} > \qty{4}{dBm},\\[0.5ex]
\gamma^{\star,\mathrm{dB}}, 
& \text{otherwise}.
\end{cases}
\end{equation}
In words, the target \gls{sdr} is maintained only when the signal at the input of the \gls{vga} has an average power within the dynamic range of the \gls{vga}. If the input power is outside the dynamic range, the \gls{sdr} in dB becomes a linear function (with slope $1$) of the average received power in \qty{}{dBm}.
\section{Lower Bounds on the Achievable Rates for a Simpler Linear Channel Model}
\label{sec:section3}
Since \gls{mmse} channel estimates are not available for the input-output relation in~\eqref{ZkBBP2}, it is necessary to derive achievable-rate expressions that hold for arbitrary channel estimation algorithms.
In this section, by focusing on a simpler linear channel model, we provide such
achievable-rate bounds. In Section~\ref{sec:sec4},
we then particularize these bounds  to \gls{dmimo} systems with 1-bit \gls{rof} fronthaul using the model introduced in
Section~\ref{sec:system-model}.
\subsection{A Simpler Block-Fading Model}
We consider a block-fading model in which the coherence time of the channel corresponds to the transmission of a block of $n\sub{c}~=~n\sub{p}~+~n\sub{d}$  symbols per subcarrier, where $n\sub{p}$ pilots are used to compute an estimate of the channel, used to construct a linear combiner, and $n\sub{d}$ denotes the number of data symbols.
Note that we do not make any assumption on the selected channel estimation algorithm or the linear
combiner.

To obtain our bounds, which are based on a random-coding argument that involves Gaussian codebooks,
we assume that each codeword spans $n\sub{b}$ independently fading blocks and then let $n\sub{b}\to\infty$.

Using a modified version of the notation used in Section~\ref{sec:system-model}, which accounts for
the block-fading nature of the system model considered in this section via the addition of the two
additional indices $\ell$ and $j$, we express the $j$th received symbol for user $u$ on subcarrier~$k$
of block $\ell$ after spatial processing at the \gls{cpu} as
\begin{align}
	\label{eq: io relation block fading}
	\xx{y}{j,u}{k}{\ell} =  g^{(\ell)}_{u,k}  \, \xx{\bar{s}}{j,u}{k}{\ell}  + & \sum_{v\neq u} g^{(\ell)}_{v,k} \, \xx{\bar{s}}{j,v}{k}{\ell} + \xx{z}{j,u}{k}{\ell}.
\end{align}
Here, $g^{(\ell)}_{u,k}$ is the effective channel for user $u$ (after spatial processing) and $\xx{z}{j,u}{k}{\ell}$ is the additive noise.
We assume that the effective channel $g^{(\ell)}_{u,k}$ is a function of the actual channel $\bar{\matH}^{(\ell)}_k \in \complexset^{B\times U}$ and its estimate, denoted by $\hat{\matH}^{(\ell)}_k$.
For example, each $g^{(\ell)}_{u,k}$ may be obtained by projecting the $u$th column of $\bar{\matH}^{(\ell)}_k$ onto a combiner vector computed on the basis of  $\hat{\matH}^{(\ell)}_k$.
Throughout, we assume that, for every fixed pair $(u,k)$, the additive noise terms $\{z^{(\ell)}_{j,u,k}\}$ have zero mean, are \iid across~$\ell$, identically distributed (but not necessarily independent) across~$j$ even when $\{\bar{\matH}^{(\ell)}_k\}$ is given, and uncorrelated with $\{\bar{s}^{(\ell)}_{j,u,k}\}$. Note that we do not require $\{z^{(\ell)}_{j,u,k}\}$ to be Gaussian or uncorrelated across~$u$ or~$k$.
Furthermore, we denote the variance of each $z^{(\ell)}_{j,u,k}$ as $\sigma^{2}_{u,k}$.
Finally, we assume that $\{g^{(\ell)}_{u,k}\}$ are independent of the transmit symbols $\{\bar{s}^{(\ell)}_{j,u,k}\}$.
It is important to highlight that we have made no assumptions on the quality of the
channel estimates, on the method used to obtain them, or on the joint
distribution of the pairs $\{\bar{\matH}^{(\ell)}_{k},
	\hat{\matH}^{(\ell)}_{k}\}$.\footnote{In particular, for some of the results in this section, we do
	not even require the channel to be Gaussian.}
This is different from most analyses available in the literature, in which bounds on the
achievable rates are derived under the simplifying assumptions that $\hat{g}^{(\ell)}_{u,k}$ is a linear function of
$\bar\matH^{(\ell)}_{k}$ (as in the combiner example we have just provided), that each $\hat{\matH}^{(\ell)}_{k}$ is
a \gls{mmse} estimate of the corresponding $\bar{\matH}^{(\ell)}_{k}$, and that $\{\bar\matH^{(\ell)}_{k}, \hat{\matH}^{(\ell)}_{k}\}$ are jointly Gaussian
(see,~e.g.,~\cite[Thm.~4.1]{Bjornson_hardware}).

%

Since our bounds hold for Gaussian codebooks, in the remainder of this section we assume, consistently with the model in Section~\ref{sec: 2A}, that each
symbol $\bar{s}^{(\ell)}_{j,u,k}$ is drawn independently over users,
subcarriers, data symbols, and fading blocks, from a $\jpg(0,\rho_{u,k})$ distribution.

We start by presenting in Section~\ref{sec:ml-bound} a reverse-channel bound, which is
obtained by assuming that the receiver uses \gls{ml} decoding to recover the signal
transmitted by user $u$.
As we shall see, this bound is in general difficult to evaluate numerically,
apart from the special case of linear combining, \gls{mmse} channel estimation, and jointly Gaussian $\{\bar\matH^{(\ell)}_{k}, \hat{\matH}^{(\ell)}_{k}\}$ (see
Section~\ref{sec:ml-bound-mmse-gaussian}), for which it reduces to that reported
in~\cite[Thm.~4.1]{Bjornson_hardware}.
Another case in which numerical computation is feasible is discussed in Section~\ref{sec:uatf}. In
this case, the receiver is given the additional side information $\left\{\Ex{}{g^{(\ell)}_{v,k}}\right\}$, and $\hat\matH^{(\ell)}_{k}$ is not used when decoding.
The resulting bound reduces to the
so-called \gls{uatf} bound~\cite[Thm.~4.4]{Bjornson_hardware}.
%

In Section~\ref{sec:gmi}, we present a second bound, which is expressed in terms of the
so-called \gls{gmi} for mismatched decoding~\cite{ganti00_11a} and pertains to the scenario in
which the decoder does not perform \gls{ml} detection, but operates according to the
suboptimal but more practical \gls{snn} decoding rule.
The main advantage of this bound is that it can be evaluated efficiently
without any simplifying assumptions.

\subsection{The Reverse-Channel Bound}\label{sec:ml-bound}
Our first result pertains to a bound achievable with a receiver that operates according to the optimal single-user \gls{ml} decoding rule.
Let $\bar{\vecs}_{u} \in \complexset^{n\sub{b}n\sub{d}S}$ be a vector containing all symbols $\{\bar{s}^{(\ell)}_{j,u,k}\}$ transmitted by user $u$.
Let $\vecy_{u}$ be the corresponding vector of received samples $\{y^{(\ell)}_{j,u,k}\}$ according
to~\eqref{eq: io relation block fading}.
Let also $\bar\matH^{(\ell)}=\{\bar \matH^{(\ell)}_{k}, k \in \setS \}$ and $\hat\matH^{(\ell)}=\{\hat \matH^{(\ell)}_{k}, k \in \setS \}$.

With \gls{ml} decoding, the receiver selects the codeword
$\hat{{\vecs}}_{u}$ in the
codebook $\setC$ used by user $u$ that has the highest likelihood.
Mathematically,
\begin{equation}\label{eq:ml-decoder}
	\hat{{\vecs}}_{u} = \argmax_{\tilde{{\vecs}} \in \setC} \mathbb{P}(\vecy_{u} \given \tilde{{\vecs}},
	\{\hat\matH^{(\ell)}\}).
\end{equation}
We are now ready to state our first bound.

\begin{thm}\label{thm:ergodic-ml-decoder}
	For the block-fading model~\eqref{eq: io relation block fading}, the rate achievable by user $u$ with \iid Gaussian codebooks and single-user
	\gls{ml} decoding is lower-bounded by
	\begin{IEEEeqnarray}{rCl}
      R\supp{rc}_{u}
      &=& \frac{\nd}{n\sub{c}S}\sum_{k\in\setS}
      \Ex{\hat{\matH}}{
          \log\lefto(1+
          \frac{\rho_{u,k} \abs{\Ex{}{g_{u,k}\given \hat{\matH}}}^{2}}
          {\Delta_{u,k}(\hat{\matH})}\right)},
      \IEEEeqnarraynumspace
      \label{eq:ml-bound}
  \end{IEEEeqnarray}
	where
	\begin{multline}
			\Delta_{u,k}(\hat{\matH})
			  = \sum_{v=1}^{U}\rho_{v,k}
			\Ex{}{\abs{g_{v,k}}^{2} \given \hat{\matH}}                        \\
			  \quad{} - \rho_{u,k} \abs{\Ex{}{g_{u,k}\given \hat{\matH}}}^{2}
			+ \Ex{}{\abs{z_{u,k}}^{2}\given \hat{\matH}}.
		\label{eq:ml-bound-denominator}
	\end{multline}
\end{thm}
Here, $\hat\matH$ is distributed as one of the $\hat\matH^{(\ell)}$.
Similarly, $g_{u,k}$ and $z_{u,k}$ are conditionally distributed as one of the $g_{u,k}^{(\ell)}$ and
$z_{j,u,k}^{(\ell)}$ in~\eqref{eq: io relation block fading}.
\begin{IEEEproof}
	See Appendix~\ref{app:proof-ml-bound}.
\end{IEEEproof}
Note that the bound in~\eqref{eq:ml-bound} is difficult to evaluate, even numerically,
because of the conditional expectations in~\eqref{eq:ml-bound} and~\eqref{eq:ml-bound-denominator}.
We discuss next two special cases in which this issue can be overcome.

\subsubsection{Special Case---Linear Combiners, Gaussian Channels, and \gls{mmse} Channel
	Estimation}\label{sec:ml-bound-mmse-gaussian}
Assume that for each $u=1,\dots,U$ and $k\in \setS$, the effective channel $g_{u,k}$ in~\eqref{eq:ml-bound} is a linear function of the actual channel $\bar\matH_{k}$, distributed as one of the $\bar\matH_{k}^{(\ell)}$.
Specifically, assume that $g_{u,k}= \herm{\vecw}_{u,k}\bar\vech_{u,k}$, where $\bar\vech_{u,k}$ is the $u$th column of $\bar\matH_k$ and the spatial combiner  $\vecw_{u,k}$ is a function of the channel estimates $\hat{\matH}$.
Let $\hat{g}_{u,k}= \herm{\vecw}_{u,k}\hat{\vech}_{u,k}$, where $\hat{\vech}_{u,k}$ is the $u$th column of $\hat\matH_{k}$ and assume that $\Ex{}{\bar\vech_{u,k}\given
		\hat{\matH}}=\hat{\vech}_{u,k}$.
Note that this assumption, which holds when $\hat{\vech}_{v,k}$ is obtained via \gls{mmse}
estimation, implies that $\Ex{}{g_{v,k}\given \hat{\matH}}= \hat{g}_{v,k}$, $\forall v,k$.
Let us finally assume that $\{\bar\matH,\hat\matH\}$ are jointly Gaussian distributed, which holds
if $\bar\matH$ is Gaussian  distributed and \gls{mmse} estimation is used.
In this case,
\begin{IEEEeqnarray}{rCl}
	\Ex{}{\abs{g_{v,k}}^{2}\given \hat{\matH}} &=& \Ex{}{\abs{g_{v,k}- \hat{g}_{v,k} + \hat{g}_{v,k}}^{2}
		\given \hat{\matH}}\\
	&=& \Ex{}{\abs{g_{v,k}-\hat{g}_{v,k}}^{2} \given \hat{\matH}}
	+ \abs{\hat{g}_{v,k}}^{2}\label{eq:variance-simplification-1} \\
	&=& \Ex{}{\abs{g_{v,k}-\hat{g}_{v,k}}^{2}}
	+ \abs{\hat{g}_{v,k}}^{2} \label{eq:variance-simplification-2}.
\end{IEEEeqnarray}
Here,~\eqref{eq:variance-simplification-1} follows because $\Ex{}{g_{v,k}\given \hat{\matH}}=
	\hat{g}_{v,k}$, and~\eqref{eq:variance-simplification-2} follows from the \gls{mmse} estimation assumption.

To summarize, under the assumptions listed in this section, the bound in~\eqref{eq:ml-bound}
simplifies to
\begin{IEEEeqnarray}{rCl}
	R\supp{rc}_{u}
	&=& \frac{\nd}{n\sub{c}S}\sum_{k\in\setS}\Ex{\hat{\matH}}{
		\log\lefto(1+
		\frac{\rho_{u,k} \abs{\hat{g}_{u,k}}^{2}}
		{\Delta_{u,k}\supp{mmse}(\hat{\matH})}\right)},
	\label{eq:ml-bound-mmse-gaussia}
\end{IEEEeqnarray}
where
\begin{multline}
		\Delta_{u,k}\supp{mmse}(\hat{\matH})
		  = \sum_{v\neq u}\rho_{v,k} \abs{\hat{g}_{v,k}}^{2}
		+ \sum_{v=1}^{U}\rho_{v,k}
		\Ex{}{\abs{g_{v,k}- \hat{g}_{v,k}}^{2}}                  \\
		  \quad{} + \Ex{}{\abs{z_{u,k}}^{2}\given \hat{\matH}}.
\end{multline}

This bound, which is much simpler to evaluate numerically than~\eqref{eq:ml-bound},
coincides essentially with the one reported in~\cite[Thm.~4.1]{Bjornson_hardware} (see
also~\cite[Thm.~4.0.2]{ganti00_11a} for an earlier result of the same flavor).

\subsubsection{Special Case---the \gls{uatf} Bound}\label{sec:uatf}
Another case in which the bound in~\eqref{eq:ml-bound} lends itself to efficient numerical
evaluation is when the channel estimates $\hat{\matH}$ are ignored (forgotten) during the decoding
process, but the receiver has perfect knowledge of the mean $\{\Ex{}{g_{v,k}}\}$ of the effective channels.
Such an assumption is reasonable in systems in which the receiver is equipped with a large number
of antennas.
Indeed, in such a scenario, the effective channel after spatial combining tends to
concentrate around its mean value---a phenomenon commonly referred to as channel
hardening~\cite{Bjornson_hardware}.
For this scenario,~\eqref{eq:ml-bound} simplifies to
\begin{IEEEeqnarray}{rCl}
	\label{eq: uatfbound}
	R\supp{uatf}_{u}
	&=& \frac{\nd}{n\sub{c}S}\sum_{k\in\setS}\log\lefto(1 +
	\frac{\rho_{u,k} \abs{\Ex{}{g_{u,k}}}^{2}}
	{\Delta_{u,k}\supp{uatf}}\right),
\end{IEEEeqnarray}
where
\begin{IEEEeqnarray}{rCl}
	\label{eq: deltauatf}
	\Delta_{u,k}\supp{uatf}
	&=& \sum_{v=1}^{U}\rho_{v,k} \Ex{}{\abs{g_{v,k}}^{2}}
	- \rho_{u,k} \abs{\Ex{}{g_{u,k}}}^{2}
	+ \sigma_{u,k}^{2}. \IEEEeqnarraynumspace
\end{IEEEeqnarray}
This bound, which coincides with the one reported in~\cite[Thm.~4.4]{Bjornson_hardware}, is
usually referred to in the massive \gls{mimo} literature as the \gls{uatf} bound. This
bounding technique was introduced in the massive-\gls{mimo} context in~\cite{ngo13-04a}.

\subsection{The \gls{gmi} Bound}\label{sec:gmi}
We provide next a lower bound on the rates achievable with a decoder that operates
according to the \gls{snn} decoding rule.
Specifically, after relying on  the  $\{\hat{\matH}^{(\ell)}\}$ to compute an estimate
$\hat{g}^{(\ell)}_{u,k}$ of the effective channel gain $g^{(\ell)}_{u,k}$ for the intended user $u$, the decoder selects the
codeword $\hat{\vecs}_{u}$ in the
codebook~$\setC$ of user~$u$ that is closest to the received signal, after each element in the
codeword is scaled by the corresponding estimated channel gain.
Mathematically,
\begin{equation}\label{eq:snn-decoder}
	\hat{{\vecs}}_{u}
	= \argmin_{\tilde{{\vecs}} \in \setC}
	\sum_{\ell=1}^{n\sub{b}} \sum_{j=1}^{n\sub{d}} \sum_{k\in \setS} \abs{y_{j,u,k}^{(\ell)}-
		\hat{g}_{u,k}^{(\ell)} \tilde{s}^{(\ell)}_{j,u,k}}^{2},
\end{equation}
where $\tilde{{\vecs}}$ contains all the $\tilde{s}^{(\ell)}_{j,u,k}$.
In this decoding rule, the receiver treats each channel estimates $\hat{g}^{(\ell)}_{u,k}$ as
perfect and the effective noise, which includes $z_{j,u,k}^{(\ell)}$ and the multiuser interference term
$\sum_{v\neq u} g^{(\ell)}_{v,k} \, \xx{\bar{s}}{j,v}{k}{\ell}$, as Gaussian.
Indeed, when $\hat{g}^{(\ell)}_{u,k}=g^{(\ell)}_{u,k}$ and the effective noise is Gaussian, this decoding rule
coincides with the single-user \gls{ml} rule~\eqref{eq:ml-decoder}.
Constraining the decoder to operate according to~\eqref{eq:snn-decoder} is of practical
interest, since the information-theoretic rates computed under~\eqref{eq:snn-decoder} can be
achieved by a receiver in which the channel estimates are treated as perfect and a capacity-achieving AWGN
channel code is used---an approach widely adopted in practical implementations.
In the next theorem, we provide a characterization of such rates.
\begin{thm}\label{thm:snn-decoder}
	For the block-fading model~\eqref{eq: io relation block fading}, the rate achievable by user $u$ with \iid Gaussian codebooks and
	\gls{snn} decoding~\eqref{eq:snn-decoder} is lower-bounded by
	\begin{IEEEeqnarray}{rCl}
		R_{u}\supp{snn}
		&=& \frac{\nd}{n\sub{c}S}\sup_{\lambda>0}\sum_{k\in\setS} \Biggl[
			-\lambda \Gamma_{u,k} \log e\nonumber\\
			&&\hspace{1.5em}
			+ \lambda \Ex{}{\frac{\Psi_{u,k}}{1+\lambda\rho_{u,k} \abs{\hat{g}_{u,k}}^{2}}}
			\log e \nonumber\\
			&&\hspace{1.5em}
			+ \Ex{}{\log(1+\lambda \rho_{u,k}\abs{\hat{g}_{u,k}}^{2})}
			\Biggr]
		\label{eq:gmi-example}
	\end{IEEEeqnarray}
	where 	%
	\begin{IEEEeqnarray}{rCl}
		\Gamma_{u,k}
		&=& \rho_{u,k} \Ex{}{\abs{\hat{g}_{u,k}-g_{u,k}}^{2}} \nonumber\\
		&&+ \sum_{v\neq u}\rho_{v,k}\Ex{}{\abs{g_{v,k}}^{2}}
		+ \sigma_{u,k}^{2},\\
		\Psi_{u,k}
		&=& \rho_{u,k}\abs{g_{u,k}}^{2}
		+ \sum_{v\neq u}\rho_{v,k}\Ex{}{\abs{g_{v,k}}^{2}}
		+ \sigma^{2}_{u,k}.
	\end{IEEEeqnarray}
	All expectations are over the actual channels
	$\bar{\matH}$ and their estimates $\hat\matH$.\footnote{We use the same notation convention as in
		Theorem~\ref{thm:ergodic-ml-decoder}.}

\end{thm}
\begin{IEEEproof}
	See Appendix~\ref{app:proof-snn-decoder}.
\end{IEEEproof}

\section{Application to the \gls{dmimo} with 1-Bit \Gls{rof} System Model}
\label{sec:sec4}
In this section, we apply the lower bounds on the achievable rate derived in Section~\ref{sec:section3} to \gls{dmimo} systems with 1-bit \gls{rof} fronthaul using the input-output relation obtained in~\eqref{eq: zkbbBussgangdecomposition}. Specifically, we consider a modified version of~\eqref{eq: zkbbBussgangdecomposition} where we replace $\bar{\vecz}_k\supp{bb}$ with $\vecz_{j,k}^{\mathrm{bb},(\ell)}$ to indicate explicitly the dependence on the symbol index $j$ and fading-block index $\ell$. All other quantities in~\eqref{eq: zkbbBussgangdecomposition} are modified accordingly. Similarly to Section~\ref{sec:section3}, we assume that, within each fading block, the first $n\sub{p}$ symbols are pilot symbols used by the \gls{cpu} to obtain an estimate $\hat{\matH}_k^{(\ell)}$ of the channel matrix  $\bar{\matH}_k^{(\ell)}$. Let $\matB_k^{(\ell)}~\in~\complexset^{U \times B}$ be the combining matrix for subcarrier $k$ in block $\ell$, and let $\tp{(\vecb_{u,k}^{(\ell)})}$ indicate the $u$th row of $\matB_k^{(\ell)}$.

To decode the $j$th symbol for user~$u$ on subcarrier~$k$ of block~$\ell$, the \gls{cpu} computes 
\begin{equation}
	\label{eq: yjkell}
	y_{j,u,k}^{(\ell)} = \tp{(\vecb_{u,k}^{(\ell)})} \vecz_{j,k}^{\mathrm{bb},(\ell)}.
\end{equation}
It follows from~\eqref{eq: zkbbBussgangdecomposition} that~\eqref{eq: yjkell} can be expressed as~\eqref{eq: io relation block fading} by setting
\begin{equation}
	\label{eq: true effictive channel 1bit}
	g_{u,k}^{(\ell)} = \tp{(\vecb_{u,k}^{(\ell)})} \matG^{(\ell)} \bar{\vech}_{u,k}^{(\ell)},
\end{equation}
where $\bar{\vech}_{u,k}^{(\ell)}$ denotes the $u$th column of $\bar{\matH}_k^{(\ell)}$, and 
\begin{align}
	\label{eq: effective noise 1bit}
	n_{j,u,k}^{(\ell)} = \tp{(\vecb_{u,k}^{(\ell)})} \lefto[ \matG^{(\ell)} (\bar{\vecw}_{j,k}^{(\ell)} + \bar{\vecd}_{j,k}^{(\ell)}) + \bar{\vece}_{j,k}^{(\ell)} \right].
\end{align}
It is also useful to identify the estimate $\hat{g}_{u,k}^{(\ell)}$ of the effective channel as 
\begin{equation}
	\label{eq: estimated effictive channel 1bit}
	\hat{g}_{u,k}^{(\ell)} = \tp{(\vecb_{u,k}^{(\ell)})} \est{\matG}^{(\ell)} \hat{\vech}_{u,k}^{(\ell)},
\end{equation}
where
\begin{equation}
	\est{\matG}^{(\ell)}
	= \sqrt{\frac{2}{\pi}}\,
	\diag\left(\est{\matR}_q^{(\ell)}[0]\right)^{-1/2},
\end{equation}
and 
\begin{equation}
	\est{\matR}_q^{(\ell)}[0]
	= \frac{1}{N}\,\Re\Big\{\sum_{k \in \setS}
	(\est{\matH}_k^{(\ell)} \matP_k \herm{(\est{\matH}_k^{(\ell)})} + N_0 \matI_B)\Big\}
	+ \frac{\matD}{2}.
\end{equation}
Finally, the conditional noise variance $(\sigma_{u,k}^{(\ell)})^2$ given $\bar{\matH}^{(\ell)}$ and  $\hat\matH^{(\ell)}$ is
\begin{IEEEeqnarray}{rCl}
  (\sigma_{u,k}^{(\ell)})^2
  &=&
  \tp{(\vecb_{u,k}^{(\ell)})}
  \!\Big(
  \matG^{(\ell)} (\matD + N_0 \matI_B)\matG^{(\ell)}
  \nonumber\\
  &&\quad{}
  + \matC_{\bar{\vece}_k}^{(\ell)}
  \Big)
  \conj{(\vecb_{u,k}^{(\ell)})} ,
  \end{IEEEeqnarray}
and $\sigma_{u,k}^2$ in~\eqref{eq: deltauatf} can be obtained by averaging over the pair $\lefto( \bar{\matH}^{(\ell)}, \hat{\matH}^{(\ell)}\right)$.
\section{Simulation Results}\label{sec:simulation-results}
Unless otherwise stated, in all simulations we use the \gls{blmmse} channel estimator in~\eqref{hBLM1} and the \gls{blmmse} combiner~\cite{hu24-11a}.
Furthermore, each \gls{ue} is constrained to use the same transmit power over all 
subcarriers. In all subsections except the power-control study in Section~\ref{sec:power-control}, the transmit power is also the same across
\glspl{ue}. Specifically, $\matP_k=\rho\matI_U$ for all $k\in\setS$, with $\rho = \qty{20}{dBm}$. The power of the additive noise is $N_0 = \qty{-94}{dBm}$. Furthermore, we set $n\sub{c} = \num{2000}$, and use the same transmit power and the same dither power matrix $\matD$ in both the pilot and the data transmission phase. Finally, we define the pathloss coefficient  $\beta_{b,u}$ in~\eqref{eq: pathloss_label} as~\cite[pp.~67--68]{3GPP_TR_36_814}
\begin{align}
	\beta_{b,u}= -37.6\log_{10}(d_{b,u})-\qty{35.3}{dB},
\end{align}
where $d_{b,u}$ is the three-dimensional distance between \gls{ap} $b$ and \gls{ue} $u$ in meters. 
\subsection{The Effect of Oversampling Rate}
\label{sec: VA}
We start by investigating the effect of the oversampling rate $N/S$ on the reverse-channel, \gls{uatf}, and \gls{gmi} bounds. We set $U=B=1$, $S=1$, $\np=1$, since computing the reverse-channel bound in~\eqref{eq:ml-bound} is numerically challenging for higher values of these parameters. Indeed, the computational complexity of the reverse-channel bound scales as $\mathcal{O}\left(S\,2^{n_{\mathrm p}BN}\right)$. Furthermore,
we set the distance between the \gls{ue} and the \gls{ap} to $d=\qty{36.7}{m}$, and $W = \qty{960}{MHz}$, which allows us to consider values of $N$ as small as $7$ while avoiding spectral aliasing. For each value of $N/S$, the \gls{sdr} is optimized via grid search, and the resulting rates are plotted in Fig.~\ref{fig: ratesvsn}.

We observe that the reverse-channel bound yields the largest rate estimates. This is expected since this bound pertains to the optimal single-user \gls{ml} decoding rule. However, it can only be evaluated for small oversampling rates: in Fig.~\ref{fig: ratesvsn}, we were able to evaluate it only up to $N/S=11$, which is far below the oversampling rates used in current testbeds ($f\sub{s} / W = N/S = 250$ in~\cite{aabel24-10a,aabel25-11a}). We also observe that the \gls{gmi} bound is above the \gls{uatf} bound for this single-\gls{ap} setup. This is reasonable, since in this setup the instantaneous effective channel does not concentrate around its mean; so using $\Ex{}{g_{u,k}}$ instead of $\hat{g}_{u,k}$ is suboptimal.

The dashed lines show the asymptotic values ($N \to \infty$) for the \gls{uatf} and \gls{gmi} achievable rates. These values are obtained by computing the rates at a very high oversampling rate. Specifically, we set $N/S = \num{40500}$. As shown in Fig.~\ref{fig:nmse-vs-n}, further increasing $N/S$ would result in a marginal reduction of the channel estimation error.

Our results indicate that at an oversampling rate of $N/S= 250$, which is the one used in the testbeds developed in~\cite{aabel24-10a,aabel25-11a}, the \gls{uatf} bound is around $28\%$ below its asymptotic value. Similarly, the \gls{gmi} bound is $37\%$ below its asymptotic value. This suggests that increasing the oversampling rate $N/S$ beyond $250$ can improve the rates substantially.

\begin{figure}
	\centering
	\includegraphics[width=0.85\linewidth]{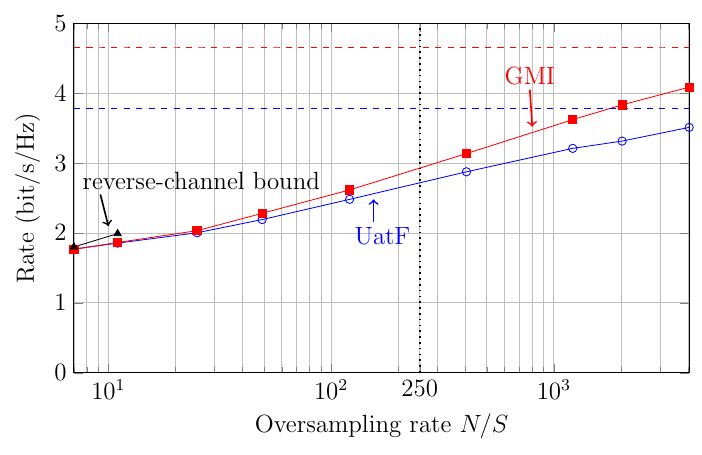}
	\caption{\gls{uatf}, \gls{gmi}, and reverse-channel bounds for different values of $N/S$. The dashed lines show the asymptotic rates.}
	\label{fig: ratesvsn}
\end{figure}

\begin{figure}
	\centering
	\includegraphics[width=0.85\linewidth]{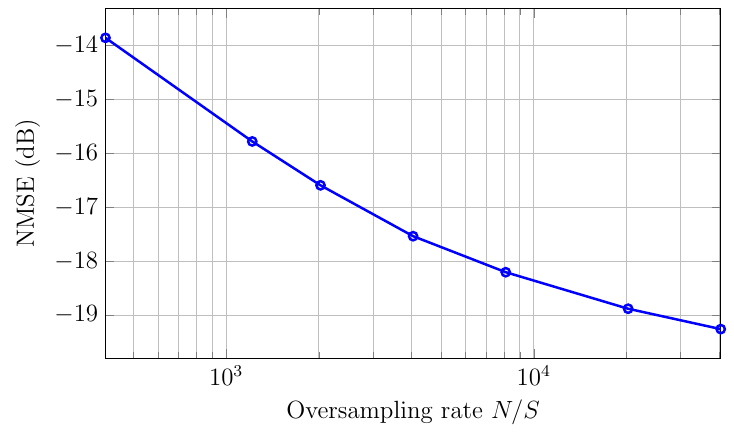}
	\caption{Channel-estimation \gls{nmse} as a function of the oversampling rate $N/S$, for the system parameters specified in Section~\ref{sec: VA}.}
	\label{fig:nmse-vs-n}
\end{figure}

\subsection{Realistic System Parameters and Optimization of the \gls{sdr}}
\label{sec: VB}
We now set the system parameters to match the ones used in the testbed~\cite{aabel24-10a,aabel25-11a}, i.e., $W=\qty{100}{MHz}$, $S=9$, and $\np=2U$.

Throughout this section, we consider the deployment scenario depicted in Fig.~\ref{fig:dmimo-topologies}.
The service area is a $W_{\mathrm r}\times W_{\mathrm r}$ square with $W_{\mathrm r}=\qty{100}{m}$. The $B=16$ \glspl{ap} are placed on a regular $4\times4$ grid at a height of $\qty{10}{m}$, and the \glspl{ue} are at ground level. In the single-user case, the \gls{ue} is placed at the center of the service area (see Fig.~\ref{fig:dmimo-topology-U1}). For the multiuser case, the four \glspl{ue} are placed as shown in Fig.~\ref{fig:dmimo-topology-MU}. 
\begin{figure}[t]
	\centering
	\begin{subfigure}[t]{0.49\columnwidth}
		\centering
		\includegraphics[width=\linewidth]{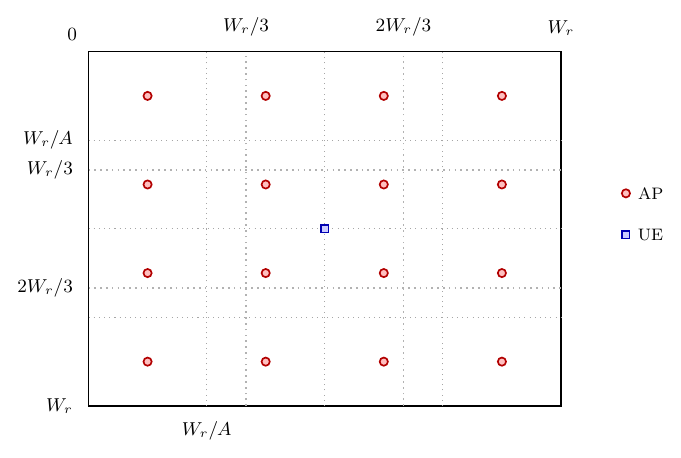}
		\caption{$U=1$, $B=16$}
		\label{fig:dmimo-topology-U1}
	\end{subfigure}
	\hfill 
	\begin{subfigure}[t]{0.49\columnwidth}
		\centering
		\includegraphics[width=\linewidth]{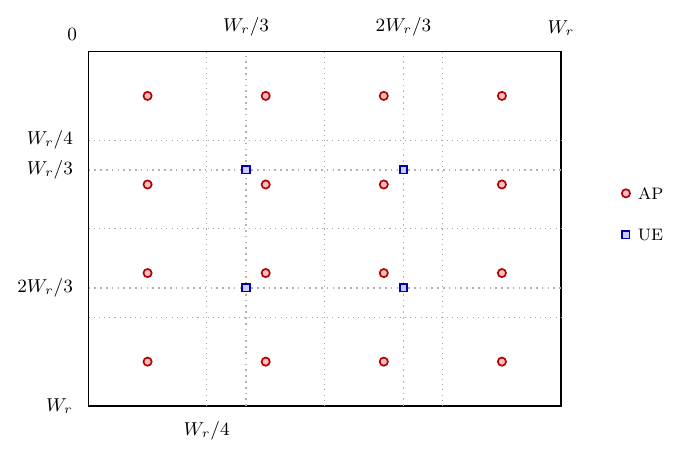}
		\caption{$U=4$, $B=16$}
		\label{fig:dmimo-topology-MU}
	\end{subfigure}
	\caption{Distributed MIMO topologies.}
	\label{fig:dmimo-topologies}
\end{figure}

\begin{figure}
	\centering
	\includegraphics[width=0.85\linewidth]{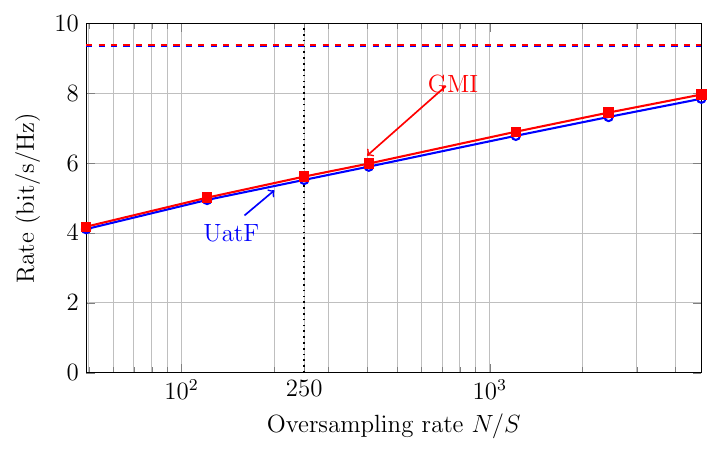}
	\caption{\gls{uatf} and \gls{gmi} bounds for different values of $N/S$ for the multiuser scenario depicted in Fig.~\ref{fig:dmimo-topology-MU}. The dashed lines show the asymptotic rates.}
	\label{fig: u4b16oversampling}
\end{figure}

\begin{figure}
	\centering
	\includegraphics[width=0.85\linewidth]{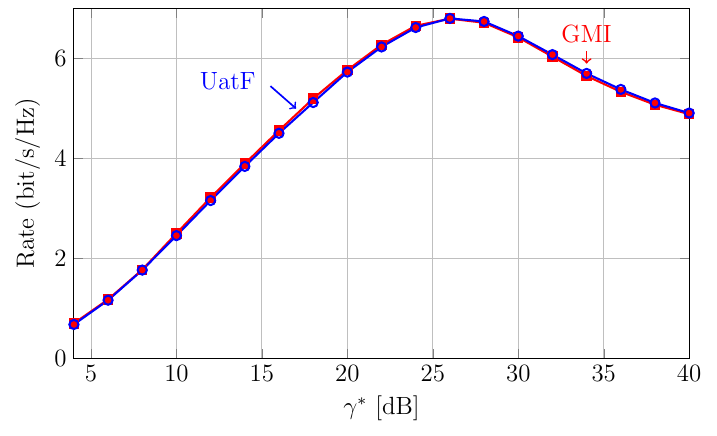}
	\caption{\gls{uatf} and \gls{gmi} bounds for $U=1$ and $B=16$.}
	\label{fig: u1b16boundstestbed}
\end{figure}

\begin{figure}
	\centering
	\includegraphics[width=0.85\linewidth]{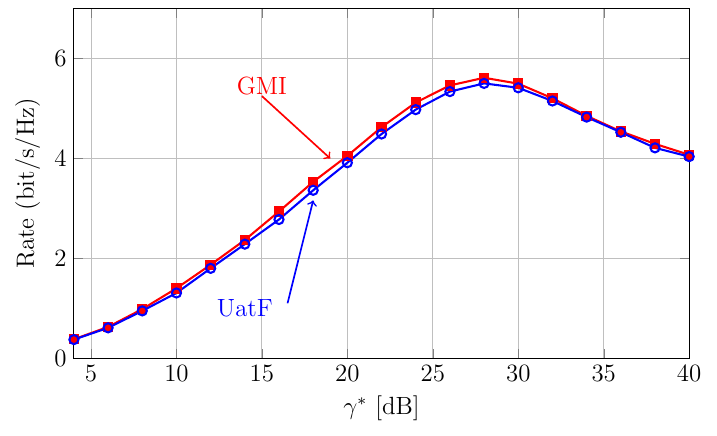}
	\caption{\gls{uatf} and \gls{gmi} bounds for $U=4$ and $B=16$.}
	\label{fig: u4b16boundstestbed}
\end{figure}

We investigate the impact of oversampling rate $N/S$ for the multiuser case in Fig.~\ref{fig: u4b16oversampling}. We observe that, at the oversampling rate of $N/S = 250$, both the \gls{uatf} and the \gls{gmi} bounds are about $40\%$ lower than their asymptotic values. In agreement with the observations from Fig.~\ref{fig: ratesvsn}, this suggests that, also for this more realistic scenario, increasing the oversampling rate $N/S$ beyond $250$ can improve the achievable rates significantly. The reverse-channel bound cannot be evaluated for the system parameters considered in this section.

 Next, we set $N/S = 250$, and investigate the sensitivity of the achievable rate to the target \gls{sdr} $\gamma^\star$, for the ideal \gls{vga} case, i.e., \gls{vga} with infinite dynamic range. We will discuss the impact of the finite dynamic range in Section~\ref{sec:power-control}.
 
In Figs.~\ref{fig: u1b16boundstestbed} and~\ref{fig: u4b16boundstestbed} we plot the \gls{uatf} and the \gls{gmi} bounds as a function of the \gls{sdr} $\gamma^\star$. 

We observe from Figs.~\ref{fig: u1b16boundstestbed} and~\ref{fig: u4b16boundstestbed} that both bounds are sensitive to $\gamma^\star$, which highlights the need to correctly tune the \gls{sdr} in  \gls{dmimo} systems with 1-bit \gls{rof} fronthaul. For $U=1$, both the \gls{uatf} and \gls{gmi} bounds are maximized at $\gamma^\star=\qty{26}{dB}$, with peak rates of $\qty{6.81}{bit/s/Hz}$ and $\qty{6.80}{bit/s/Hz}$, respectively. For $U=4$, both bounds are maximized at $\gamma^\star=\qty{28}{dB}$, with corresponding rates of $\qty{5.50}{bit/s/Hz}$ for \gls{uatf} and $\qty{5.62}{bit/s/Hz}$ for \gls{gmi}. Interestingly, the gap between the \gls{gmi} and the \gls{uatf} is minimal in both scenarios. In Fig.~\ref{fig: u1b16boundstestbed}, the small gap is probably due to the symmetric deployment scenario, which induces a certain degree of channel hardening. In Fig.~\ref{fig: u4b16boundstestbed}, the small gap suggests that imperfect multiuser interference cancellation, rather than the imperfect channel-state information at the decoder, is the main cause of performance degradation, together with quantization errors.

\subsection{Random User Placement}
We next use the same \qty{100}{m} by \qty{100}{m} deployment and $4\times4$ \gls{ap} grid as in Fig.~\ref{fig:dmimo-topologies}, but place the $U=4$ \glspl{ue} uniformly at random over the service area. The rates are evaluated over $\num{1000}$ random placements by fixing $\gamma^\star = \qty{28}{dB}$ as suggested by Fig.~\ref{fig: u4b16boundstestbed}. In Fig.~\ref{fig: randomuserdrop}, we show the resulting empirical \gls{cdf}. The small gap between the \gls{gmi} and \gls{uatf} bound, which is in agreement with Fig.~\ref{fig: u4b16boundstestbed}, indicates that the two bounds yield similar rate predictions for the considered system parameters when the \glspl{ue} are placed at random. Specifically, the $10$th percentiles of the empirical \gls{cdf} for the \gls{gmi} and \gls{uatf} rates are \qty{4.41}{bit/s/Hz} and \qty{4.29}{bit/s/Hz}, respectively.

\begin{figure}
	\centering
	\includegraphics[width=0.85\linewidth]{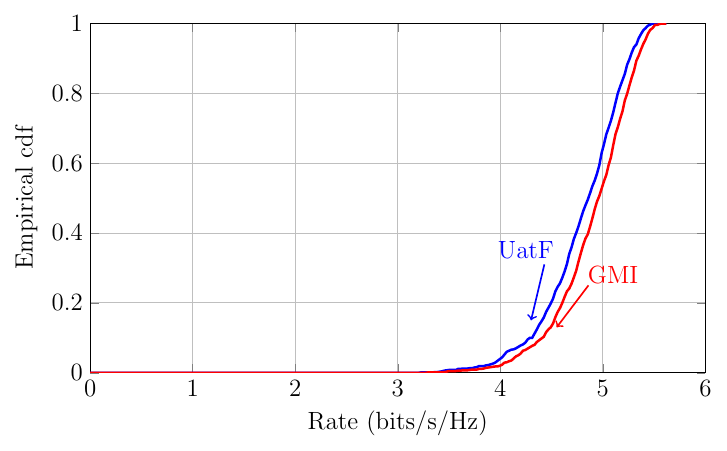}
	\caption{Empirical \gls{cdf}  for $U=4$ and $B=16$. The \gls{sdr} for data transmission is set to \qty{28}{dB}.}
	\label{fig: randomuserdrop}
\end{figure}

\subsection{Max-Min Fairness in the Presence of the \gls{vga} with Finite Dynamic Range}
\label{sec:power-control}
\paragraph*{Problem formulation}
We consider the max-min fairness problem, which involves adapting the per-\gls{ue} transmit powers so that the rate for the weakest \gls{ue} is maximized. Solving the max-min fairness problem is challenging in the \gls{dmimo} systems with 1-bit
\gls{rof} fronthaul because the transmit powers do not only affect multiuser interference. They
also determine the \gls{vga} input powers and, therefore, the \gls{sdr}, the
Bussgang gain, and the covariance matrix of the  quantization noise. As a result,
increasing the power of a given \gls{ue} does not necessarily increase its rate.

Let $\rho_u$ be the per-subcarrier power level used by \gls{ue}~$u$. We assume that $\rho_u \leq \rho\sub{max}$, where $\rho\sub{max}$ is the maximum power per subcarrier.
Using for tractability the \gls{uatf} bound as objective, we formulate max-min fairness problem as follows:
\begin{IEEEeqnarray}{rCl}
	\underset{\{\rho_{u}\}}{\mathrm{maximize}}
	&& \min_{u=1,\ldots,U} R\supp{uatf}_u(\{\rho_u\}) \nonumber\\
	\mathrm{subject\ to} \;
	&& 0 \le \rho_{u}\le \rho\sub{max}, \quad u=1,\ldots,U.
	\label{eq:power-control-maxmin}
\end{IEEEeqnarray}

Since the rate functions in~\eqref{eq:power-control-maxmin} depend on $\{\rho_u\}$ in an intricate and nonconventional way, we
use a \gls{ga} to solve~\eqref{eq:power-control-maxmin}, since this algorithm does not require convexity, differentiability, or monotonicity of the objective. In our implementation, each
chromosome of the \gls{ga} encodes the powers $\{\rho_{u}\}$, the fitness is the objective in
\eqref{eq:power-control-maxmin}, and the population is evolved using selection, crossover,
mutation, projection onto the feasible power interval, as well as elitist replacement~\cite{eiben2015introduction}.

\paragraph*{Numerical results}
We next describe numerical experiments in which we illustrate the negative impact of the finite dynamic range of the \gls{vga} on the achievable rate computed via the \gls{uatf} bound. 
Throughout this section, we consider the system parameters specified in Section~\ref{sec: VB}, with the exception of the transmit power (which we shall vary/optimize to isolate the impact of the finite dynamic range of the \glspl{vga}) and the deployment scenario.

As far as this last point is concerned, we consider the setup in Fig.~\ref{fig:agc-saturation-map}: an indoor scenario, similar to the one described in~\cite{aabel25-11a}, in which four \glspl{ap}, denoted as \gls{ap} 1, \gls{ap} 2, \gls{ap} 3, and \gls{ap} 4, located in positions $(2,2)$, $(2,6)$, $(6,2)$, $(6,6)$, respectively, provide coverage over a \qty{4}{m} $\times$ \qty{4}{m} area. We also assume that the \glspl{ap}, which are mounted at a height of \qty{0.3}{m}, serve 2 \glspl{ue} located on the ground plane, one in the middle of the coverage area, in position $(4,4)$, and one close to \gls{ap} 4, in position $(5.9, 5.9)$. We denote these two \glspl{ue} as \gls{ue} 1 and \gls{ue} 2, respectively. Finally, we set the maximum transmit power $\rho\sub{max}$ to \qty{-9.54}{dBm}.

We start by considering the case in which the \glspl{vga} have infinite dynamic range and set $\rho_u = \rho\sub{max}$, $u \in \{1,2 \}$.
For this case, the \gls{uatf} yields the following achievable rate estimates:
$R_1\supp{uatf} = \qty{4.57}{bit/s/Hz}$ and $R_2\supp{uatf} = \qty{3.93}{bit/s/Hz}$. As expected, \gls{ue} 1 enjoys a slightly larger rate than \gls{ue} 2 as it can better leverage macro diversity, whereas \gls{ue} 2 does not suffer or benefit from the proximity to \gls{ap} 4, since the \gls{vga} is able to enforce the desired \gls{sdr} $\gamma^\star = \qty{28}{dB}$. 

Let us now consider the case in which the dynamic range of the \gls{vga} is finite, according to the model given in~\eqref{eq:actual-signal-to-dither-ratio}. It turns out that when $\rho_u = \rho\sub{max}$, $u \in \{1, 2\}$, the values of $P_b\supp{lna}$ at the four \glspl{ap} are \qty{-15.4}{dBm}, \qty{-15.1}{dBm}, \qty{-20.9}{dBm}, and \qty{15.7}{dBm}. Note in particular that the \gls{vga} at \gls{ap} 4 operates outside its dynamic range, according to~\eqref{eq:actual-signal-to-dither-ratio}; so it cannot maintain the desired \gls{sdr}. This has a negative impact on the achievable rates, which drop to  $R_1\supp{uatf} = \qty{4.47}{bit/s/Hz}$ and  $R_2\supp{uatf} = \qty{3.18}{bit/s/Hz}$, respectively. This is a reduction of $2\%$ for \gls{ue} 1 but of $19\%$ for \gls{ue} 2. To mitigate the negative effect of the \glspl{vga}' dynamic range on $R_2\supp{uatf}$, we solve the max-min fairness optimization problem in~\eqref{eq:power-control-maxmin} using the \gls{ga} described earlier. This optimization algorithm yields $\rho_1 = \qty{-23.19}{dBm}$ and $\rho_2 = \qty{-19.51}{dBm}$, which result in $R_1\supp{uatf} = \qty{4.14}{bit/s/Hz}$ and $R_2\supp{uatf} = \qty{4.12}{bit/s/Hz}$. A few observations are in order. 
\begin{inparaenum}[(i)]
    \item The power levels are well below $\rho\sub{max}$. This is not surprising, since the \glspl{vga} have a \qty{45}{dB} dynamic range, and with $\rho_u = \qty{-9.54}{dBm}$ (as considered in our first two experiments), the system operates at the upper end of the dynamic range.
    \item The much lower transmit powers result in a significant increase of the rate achievable by \gls{ue} 2 at the cost of a slight reduction of the rate achievable by \gls{ue} 1.
    \item Surprisingly, at the optimal power levels, \gls{ap} 4 operates still outside its dynamic range. Indeed, the corresponding $P_b\supp{lna}$ value is \qty{5.76}{dBm}. 
\end{inparaenum}	
To shed light on this finding, we conduct an additional experiment in which we investigate how $R_1\supp{uatf}$ and $R_2\supp{uatf}$ vary as a function of $\rho_2$ when $\rho_1$ is fixed to the value~\qty{-23.19}{dBm} found by the \gls{ga}. The results are shown in Fig.~\ref{fig:ap-ablation-ue2-power}, where we also plot the rate achievable by \gls{ue} 2 when only \gls{ap} 4 is active and when only \glspl{ap} 1--3 are active. As shown in the figure, the optimal solution of the max-min fairness problem is achieved at a $\rho_2$ value above the one needed to guarantee that the \gls{vga} at \gls{ap} 4 operates within its dynamic range.

The rate curves for the case in which only a subset of \glspl{ap} are active offer a possible explanation: the slight reduction in the rate achievable when only \gls{ap} 4 is active, incurred when operating at  a $\rho_2$ level above the one needed to remain in the \gls{vga} dynamic range, is outweighed by the increase in the rate achievable when only \glspl{ap} 1--3 are active.
\begin{figure}[t]
	\centering
	\includegraphics[width=0.8\linewidth]{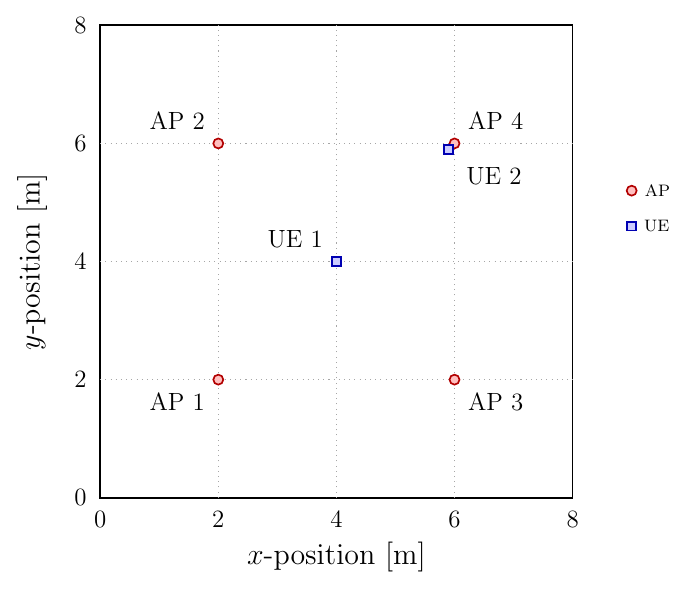}
	\caption{The selected geometry for the max-min fairness scenario.}
	\label{fig:agc-saturation-map}
\end{figure}
\begin{figure}[t]
	\centering
	\includegraphics[width=0.85\linewidth]{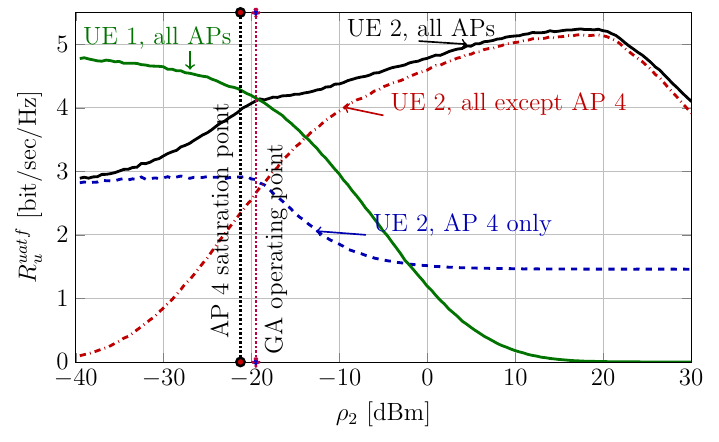}
	\caption{Achievable rates for \gls{ue} 1 and \gls{ue} 2 as a function of the transmit power of \gls{ue} 2 when the transmit power of \gls{ue} 1 is fixed to \qty{-23.19}{dBm}. We also illustrate separately the rates achievable by \gls{ue} 2 when only \gls{ap} 4 is active and when only \glspl{ap} 1--3 are active.}
	\label{fig:ap-ablation-ue2-power}
\end{figure}
\section{Conclusion}
\label{sec: conclusion}
We have derived bounds on the uplink rates achievable over a \gls{dmimo} system with 1-bit
\gls{rof} fronthaul. The bounds generalize previously available bounds in the \gls{dmimo}
literature in that they account for the specific nonlinearities present in the \gls{dmimo} systems with 1-bit
\gls{rof} fronthaul, and they hold irrespective of the algorithm used for channel estimation. In
particular, they do not require \gls{mmse} channel estimates, which are difficult to obtain in 1-bit
architectures.

The reverse-channel bound yields the largest rate estimates. Unfortunately, this bound involves the computation of conditional expectation terms that are numerically challenging to evaluate exactly for  system parameter values of practical interest. On the contrary, both the \gls{gmi} and the \gls{uatf} bound can be evaluated in practically relevant scenarios. Interestingly, our numerical experiments indicate that both bounds result in similar achievable-rate estimates.

We used the bounds to assess the impact of oversampling on the achievable rates and their sensitivity
to the chosen \gls{sdr}.
Our results indicate that using higher oversampling rates in the experiments reported in~\cite{aabel24-10a,aabel25-11a} may lead to substantial rate improvements.
Furthermore, our results point out  that the achievable rates are highly sensitive to the
\gls{sdr}, and indicate that, with Gaussian dithering, an \gls{sdr} of \qty{28}{dB} should be targeted. Finally, we show that the finite dynamic range of the \glspl{vga} at the \glspl{ap} can significantly deteriorate the rates unless \gls{ue} transmit-power optimization is performed. We have cast this problem as a max-min fairness optimization and showed that the max-min fairness algorithms available for \gls{dmimo} are not suited to our system model, since the rate function in our setup has a nonstandard dependence on the transmit power. We have also shown numerically that the optimal power allocation strategy does not necessarily result in all \glspl{vga} operating within their dynamic range, since the required reduction in transmit power may limit the benefits of macrodiversity. We leave the development of a more numerically efficient max-min fairness algorithm for future work.

\appendices

\section{Proof of Theorem~\ref{thm:ergodic-ml-decoder}}\label{app:proof-ml-bound}
We first introduce some useful additional notation. We let $\bar{\vecs}_u^{(\ell)}$ denote the vector containing all symbols transmitted by user~$u$ in block~$\ell$ and ${\vecy}_u^{(\ell)}$ the corresponding received symbols after spatial combining. 
As in Section~\ref{sec: 2A}, $\bar{\vecs}_u$ and $\vecy_u$ denote the transmitted and received symbols over the $n\sub{b}$ blocks.

An achievable rate in bit/s/Hz over the channel in~\eqref{eq: io relation block fading} can be expressed as 
\begin{align}
	R_u & =
	\lim_{\nb\to\infty}\frac{1}{\nb}
	\frac{1}{n\sub{c}S}
	I\lefto(\bar{\vecs}_u;\vecy_u
	\given \{ \hat{\matH} ^{(\ell)} \}\right) \label{eq: rc step 1}\\ 
	 &= \lim_{\nb\to\infty}\frac{1}{\nb}
	\frac{1}{n\sub{c}S}
	\sum_{\ell = 1}^{\nb} I\lefto(\bar{\vecs}_u^{(\ell)};\vecy_u^{(\ell)}
	\given \{\hat{\matH}^{(\ell)}\}\right) \label{eq: rc step 2}\\
	&= \frac{n\sub{b}}{n\sub{c}S}
	 I\lefto(\bar{\vecs}_u^{(1)};\vecy_u^{(1)}
	\given \{\hat{\matH}^{(1)}\}\right). \label{eq: rc step 3}
\end{align}

Here,~\eqref{eq: rc step 1} follows from the channel coding theorem~\cite[Prop. 2.1]{sethuraman2005capacity}, and~\eqref{eq: rc step 2} follows because the channel in~\eqref{eq: io relation block fading} is block-memoryless and the input distribution is \iid, and~\eqref{eq: rc step 3} holds because the mutual information takes the same value over each block.

The exact evaluation of~\eqref{eq: rc step 3} is complicated because the entries in $\vecy_u^{(1)}$ are not conditionally \iid. To overcome this issue, we lower-bound the mutual information in~\eqref{eq: rc step 3} as follows:
 \begin{IEEEeqnarray}{rCl}
  \IEEEeqnarraymulticol{3}{l}{
  I\lefto(\bar{\vecs}_u^{(1)};\vecy_u^{(1)}
  \given \{\hat{\matH}^{(1)}\}\right)}
  \nonumber\\
  \quad &\geq&
  \sum_{j=1}^{n\sub{d}} \sum_{k \in \setS}
  I\lefto(\bar{s}_{j,u,k}^{(1)};y_{j,u,k}^{(1)}
  \given \{\hat{\matH}^{(1)}\}\right)
  \IEEEeqnarraynumspace
  \label{eq: rc step 4}
  \\
  &=&
  n\sub{d}\sum_{k\in\setS}
  I\lefto(\bar{s}_{1,u,k}^{(1)};y_{1,u,k}^{(1)}
  \given \hat{\matH}^{(1)}\right).
  \label{eq: rc step 5}
  \end{IEEEeqnarray}
Here,~\eqref{eq: rc step 4} follows from the chain rule for mutual information and the nonnegativity of mutual information, and in~\eqref{eq: rc step 5} we used the fact that, for a given $k$, the mutual information~$I\left(\bar{s}_{j,u,k}^{(1)};y_{j,u,k}^{(1)}\given \{\hat{\matH}^{(1)}\}\right)$ takes the same value for every $j \in \{1,\dots,n\sub{d}\}$ by our assumption on $\{z_{j,u,k}^{(\ell)}\}$. Substituting~\eqref{eq: rc step 5} into ~\eqref{eq: rc step 3}, we conclude that 
\begin{equation}
	\label{eq:iwanttoboundthiswithgmi}
	R_u \geq \frac{\nd}{n\sub{c}S}\sum_{k\in\setS}
	\,
	I(\bar{s}_{1,u,k}^{(\ell)};y_{1,u,k}^{(\ell)}\given\hat{\matH}^{(\ell)})
	.
\end{equation}
In what follows, we drop the indices $\ell = 1$ and $j = 1$ to keep the notation compact. To conclude the proof, we proceed as in the proof of~\cite[Eq.~(65)]{ganti00-11b}, and lower-bound the per-subcarrier mutual information $I(\bar{s}_{u,k};y_{u,k}\given\hat{\matH})$. Specifically, let $h(\cdot)$ denote the differential entropy. Then, 
\begin{align}
	 & I(\bar{s}_{u,k}; y_{u,k}\given \hat{\matH})
	= h(\bar{s}_{u,k})
	- h(\bar{s}_{u,k}\given y_{u,k},\hat{\matH})
	\label{eq:mi-diff-entr-1}                      \\
	 & = \log(\pi e \rho_{u,k})
	- h(\bar{s}_{u,k}\given y_{u,k}, \hat{\matH})
	\label{eq:mi-diff-entr-2}                      \\
	 & \geq \log(\pi e \rho_{u,k})
	- h(\bar{s}_{u,k}-\alpha y_{u,k}\given \hat{\matH})
	\label{eq:mi-diff-entr-3}                      \\
	 & \geq \log(\pi e \rho_{u,k})
	- \Ex{\hat{\matH}}{\log\Big(\pi e
	\varepsilon_{\alpha,k}(\hat {\matH})\Big)},
	\label{eq:mi-diff-entr-4}
\end{align}
where 
\begin{equation}
	\label{eq: epsilondef}
	\varepsilon_{\alpha,k}(\hat \matH)
	=
	\Ex{}{\abs{\bar{s}_{u,k}-\alpha y_{u,k}}^{2}\given\hat{\matH}}.
\end{equation}
Here,~\eqref{eq:mi-diff-entr-1} follows from the
definition of mutual information; in~\eqref{eq:mi-diff-entr-2}, we used that
$\bar{s}_{u,k}\distas\jpg(0,\rho_{u,k})$; in~\eqref{eq:mi-diff-entr-3}, which holds for every scalar $\alpha$,
we used that conditioning does not increase differential entropy; finally,
\eqref{eq:mi-diff-entr-4}
follows because Gaussian distributions maximize differential entropy under a second-moment
constraint.
Now choose $\alpha$ so as to minimize the mean square error
$\varepsilon_{\alpha,k}(\hat {\matH})$ in~\eqref{eq: epsilondef}.
This results in
\begin{equation}
	\begin{split}
		\varepsilon_{\alpha,k}(\hat {\matH})
		 & = \rho_{u,k}
		- \frac{\rho_{u,k}^{2}\abs{\Ex{}{g_{u,k}\given\hat{\matH}}}^{2}}
		{A_{u,k}(\hat{\matH})},
	\end{split}
	\label{eq:mmse-expression}
\end{equation}
where
\begin{equation}A_{u,k}(\hat{\matH}) =  \sum_{v=1}^U \rho_{v,k}
	\Ex{}{\abs{g_{v,k}}^{2} \given \hat{\matH}}
	+\Ex{}{\abs{z_{u,k}}^{2}\given \hat{\matH}}.
\end{equation}
By substituting~\eqref{eq:mmse-expression} into~\eqref{eq:mi-diff-entr-4} and~\eqref{eq:mi-diff-entr-4} into~\eqref{eq:iwanttoboundthiswithgmi}, and by noting that 
\begin{align}
	 \Delta_{u,k}(\hat{\matH}) = A_{u,k}(\hat{\matH})- \rho_{u,k}\abs{\Ex{}{g_{u,k}\given\hat{\matH}}}^{2},
\end{align}
we obtain~\eqref{eq:ml-bound}.
\section{Proof of Theorem~\ref{thm:snn-decoder}}
\label{app:proof-snn-decoder}

We start by defining
\begin{IEEEeqnarray}{rCl}
	q_{\lambda}
	\lefto(
		\xx{y}{j,u}{k}{\ell}
		\given
		\hat{\matH}^{(\ell)},
		\xx{\bar{s}}{j,u}{k}{\ell}
	\right)
	& = & \nonumber\\
	\IEEEeqnarraymulticol{3}{l}{
		\qquad
		\exp\lefto(
			-\lambda
			\abs{
				\xx{y}{j,u}{k}{\ell}
				-
				\hat g_{u,k}^{(\ell)}
				\xx{\bar{s}}{j,u}{k}{\ell}
			}^{2}
		\right),
	}
	\label{eq:app-single-letter-metric}
\end{IEEEeqnarray}
where $\lambda > 0$ is an auxiliary parameter which originates from a Chernoff step, and will be optimized later. In~\eqref{eq:app-single-letter-metric}, we used that  $\hat{g}^{(\ell)}_{u,k}$ is a deterministic function of $\hat{\matH}^{(\ell)}$.
Using~\eqref{eq:app-single-letter-metric}, we express the \gls{snn} decoding rule in terms of the maximization of the function 
\begin{equation}
	\begin{split}
		 & q_{\lambda}
		\lefto(
		\vecy_u
		\given
		\{\hat{\matH}^{(\ell)}\},
		\bar{\vecs}_u
		\right)        \\
		 & \quad =
		\prod_{\ell = 1}^{\nb}
		\prod_{j=1}^{\nd}
		\prod_{k\in\setS}
		q_{\lambda}
		\lefto(
		\xx{y}{j,u}{k}{\ell}
		\given
		\hat{\matH}^{(\ell)},\xx{\bar{s}}{j,u}{k}{\ell}
		\right) .
	\end{split}
	\label{eq:app-product-metric}
\end{equation}
Using this metric, one can show that 
the
rate achievable by user $u$ is lower-bounded as~\cite[Eq.~(10)]{kramer_information_2023}
\begin{IEEEeqnarray}{rCl}
R_u
& \geq &
\lim_{\nb\to\infty}
\frac{1}{\nb n\sub{c} S}
\sup_{\lambda>0}
\nonumber\\[-0.25ex]
&&
\Ex{}{\log
\frac{
q_{\lambda}
\lefto(
\vecy_u
\given
\{\hat{\matH}^{(\ell)}\},
\bar{\vecs}_u
\right)
}{
\Ex{\tilde{\vecs}_u}{
q_{\lambda}
\lefto(
\vecy_u
\given
\{\hat{\matH}^{(\ell)}\},
\tilde{\vecs}_u
\right)
}
}},
\label{eq:app-gmi-full}
\end{IEEEeqnarray}
where $\tilde{\vecs}_u$ is independent of $\bar{\vecs}_u$. 
Note that 
\begin{IEEEeqnarray}{rCl}
\IEEEeqnarraymulticol{3}{l}{
	\Ex{}{\log
		\frac{
			q_{\lambda}
			\lefto(
				\vecy_u
				\given
				\{\hat{\matH}^{(\ell)}\},
				\bar{\vecs}_u
			\right)
		}{
			\Ex{\tilde{\vecs}_u}{
				q_{\lambda}
				\lefto(
					\vecy_u
					\given
					\{\hat{\matH}^{(\ell)}\},
					\tilde{\vecs}_u
				\right)
			}
		}
	}
}
\notag
\\
& = &
\sum_{\ell = 1}^{\nb}
\sum_{j = 1}^{n\sub{d}}
\sum_{k \in \setS}
\Ex{}{\iota_{j,u,k}^{(\ell)}}
\IEEEeqnarraynumspace
\label{eq:first-expansion}
\\
& = &
\nb n\sub{d}
\sum_{k \in \setS}
\Ex{}{\iota_{1,u,k}^{(1)}},
\IEEEeqnarraynumspace
\label{eq:symmetry-reduction}
\end{IEEEeqnarray}
where
\begin{IEEEeqnarray}{rCl}
\iota_{j,u,k}^{(\ell)}
& = &
\log
\frac{
	q_{\lambda}
	\lefto(
		y_{j,u,k}^{(\ell)}
		\given
		\{\hat{\matH}^{(\ell)}\},
		\bar{s}_{j,u,k}^{(\ell)}
	\right)
}{
	\Ex{\tilde{s}_{j,u,k}^{(\ell)}}{
		q_{\lambda}
		\lefto(
			y_{j,u,k}^{(\ell)}
			\given
			\{\hat{\matH}^{(\ell)}\},
			\tilde{s}_{j,u,k}^{(\ell)}
		\right)
	}
}.
\label{eq:iota-definition}
\end{IEEEeqnarray}
Here,~\eqref{eq:first-expansion} follows because  the metric $q_\lambda(\cdot)$ factorizes across $\ell$, $j$, and $k$, and because $\tilde{\vecs}_u$ has \iid entries;~\eqref{eq:symmetry-reduction} holds because, for a given $k$, the expectation in~\eqref{eq:first-expansion} takes the same value for all $j \in \{1, \dots, n\sub{d}\}$ and $\ell \in \{1, \dots, \nb\}$.

We obtain~\eqref{eq:gmi-example} by evaluating the expectation in the denominator of~\eqref{eq:symmetry-reduction} in a similar manner to~\cite[App. A]{ostman21-10y}, and then by computing in closed-form the outer expectation with respect to $\bar{s}_{j,u,k}^{(\ell)}$ and $z_{j,u,k}^{(\ell)}$.
\bibliographystyle{IEEEtran}
\bibliography{one-bit}
\end{document}